\documentclass[11pt]{article}

\usepackage{epsfig}
\usepackage{amssymb, amsmath, amsfonts, amsthm}

\usepackage{subfigure}
\usepackage{algorithmic}
\usepackage{algorithm}
\numberwithin{algorithm}{section}

\usepackage{pifont,graphicx,subfigure}
\usepackage{pstricks, pst-node}
\usepackage{color}

\usepackage{framed}
\usepackage{fullpage}

\newtheorem{theorem}{Theorem}[section]

\newtheorem{proposition}[theorem]{Proposition}
\newtheorem{lemma}[theorem]{Lemma}
\newtheorem{claim}[theorem]{Claim}
\newtheorem{observation}[theorem]{Observation}

\theoremstyle{definition}
\theoremstyle{definition}\newtheorem{definition}[theorem]{Definition}

\newcommand{\pname}[1]{{\sc #1}}

\newcommand{\comment}[1]{}
\newcommand{\QED}{\mbox{}\hfill \rule{3pt}{8pt}\vspace{10pt}\par}

\def\cost{\mathrm{cost}}
\def\deg{\mathrm{deg}}

\newcommand{\ignore}[1]{}
\newcommand{\eat}[1]{}

\newcommand{\integer}{{\mathbb Z}}
\newcommand{\real}{{\mathbb R}}

\newcommand{\squishlist}{\begin{itemize}}
\newcommand{\squishend}{\end{itemize}}

\def\danupon#1{}

\begin{document}

\title{Faster Algorithms for Semi-Matching Problems\thanks{The preliminary version of this paper appeared as \cite{FakcharoenpholLN10} in the Proceeding of the 37th International Colloquium on Automata, Languages and Programming, (ICALP) 2010.}}

\author{
Jittat Fakcharoenphol\\ {\small Kasetsart University}\\
{\tt\small jittat@gmail.com} \and
Bundit Laekhanukit\thanks{Most of the work was done while all authors were at
  Kasetsart University.}\\ \small McGill University\\
{\tt\small blaekh@cs.mcgill.ca} \and
Danupon Nanongkai\footnotemark[\value{footnote}]\\\small University of Vienna\\
{\tt\small danupon@gmail.com}
}

\date{}

\maketitle

\begin{abstract}
We consider the problem of finding \textit{semi-matching} in
bipartite graphs which is also extensively studied under various
names in the scheduling literature. We give faster algorithms for
both weighted and unweighted cases.

For the weighted case, we give an $O(nm\log n)$-time algorithm,
where $n$ is the number of vertices and $m$ is the number of edges,
by exploiting the geometric structure of the problem. This improves
the classical $O(n^3)$-time algorithms by Horn [Operations Research 1973] 
and Bruno, Coffman and Sethi [Communications of the ACM 1974].

For the unweighted case, the bound can be improved even further.
We give a simple divide-and-conquer algorithm which runs in
$O(\sqrt{n}m\log n)$ time, improving two previous $O(nm)$-time
algorithms by Abraham [MSc thesis, University of Glasgow 2003] and
Harvey, Ladner, Lov\'asz and Tamir [WADS 2003 and Journal of
Algorithms 2006]. We also extend this algorithm to solve the
\textit{Balanced Edge Cover} problem in $O(\sqrt{n}m\log n)$ time,
improving the previous $O(nm)$-time algorithm by Harada, Ono,
Sadakane and Yamashita [ISAAC 2008].
\end{abstract}

\section{Introduction}
\label{sect:intro}

In this paper, we consider a relaxation of the maximum bipartite
matching problem called \textit{semi-matching} problem, in both
weighted and unweighted cases. This problem has been previously studied
in the scheduling literature under different names, mostly known as
(non-preemptive) scheduling independent jobs on unrelated machines to
minimize flow time, or $R||\sum C_j$ in the standard scheduling
notation~\cite{Scheduling_book,Scheduling_survey,AMOBook}.

Informally, the problem can be explained by the following off-line
load balancing scenario. We are given a set of jobs and a set of machines.
Each machine can process one job at a time and it takes different
amounts of time to process different jobs. Each job also requires
different processing times if it is processed by different machines. One
natural goal is to have all jobs processed with the minimum {\em
total completion time}, or {\em total flow time}, which is the
summation of the duration each job has to wait until it is finished.
Observe that if the assignment is known, the order each machine
processes its assigned jobs is clear: It processes jobs in an
increasing order of the processing time.

To be precise, the semi-matching problem is as follows. Let
$G=(U\cup V, E)$ be a weighted bipartite graph, where $U$ is a set
of jobs and $V$ is a set of machines. For any edge $uv$, let
$w_{uv}$ be its weight. Each weight of an edge $uv$ indicates time
it takes $v$ to process $u$.
Through out this paper, let $n$ denote the number of vertices and
$m$ denote the number of edges in $G$.
A set $M\subseteq E$ is a {\em semi-matching} if each job $u \in U$
is incident with exactly one edge in $M$. For any semi-matching $M$,
we define the {\em cost} of $M$, denoted by $\cost(M)$, as follows.
First, for any machine $v\in V$, its cost with respect to a
semi-matching $M$ is
%
\[\cost_M(v)= (w_1)+(w_1+w_2)+\ldots +(w_1+\ldots
+w_{\deg_M(v)}) =\sum_{i=1}^{\deg_M(v)}(\deg_M(v)-i+1)\cdot
w_i
\]
where $\deg_M(v)$ is the degree of $v$ in $M$ and $w_1 \leq w_2 \leq
\ldots \leq w_{\deg_M(v)}$ are weights of the edges in $M$ incident
with $v$ sorted increasingly. Intuitively, this is the total
completion time of jobs assigned to $v$. Note that for the
unweighted case (i.e., when $w_e=1$ for every edge $e$), the cost of
a machine $v$ is simply $\deg_M(v)\cdot (\deg_M(v)+1)/2$.
Now, the cost of the semi-matching $M$ is simply the summation of
the cost over all machines: $$\cost(M) = \sum_{v\in V} \cost_M(v).$$
The goal is to find an {\em optimal semi-matching}, a semi-matching
with minimum cost.

\paragraph{Related works} Although the name ``semi-matching'' was
recently proposed by Harvey, Ladner, Lov\'{a}sz, and
Tamir~\cite{HLLT06}, the problem was studied as early as 1970s when
an $O(n^3)$ algorithm was independently developed by
Horn in~\cite{horn1973} and by Bruno, Coffman and Sethi
in~\cite{BrunoCS74}. Since then no progress has been made on this
problem except on its special cases and variations.
For the special case of {\em inclusive set restriction}
where, for each pair of jobs $u_1$ and $u_2$, either all neighbors
of $u_1$ are neighbors of $u_2$ or vice versa, a faster algorithm
with $O(n^2)$ running time was given by Spyropoulos and
Evans~\cite{SpyropoulosE85}. Many variations of this problem were
proved to be NP-hard, including the preemptive
version~\cite{Sitters01}, the case when there are
deadlines~\cite{su2009}, and the case of optimizing total weighted
tardiness~\cite{logendran2004}. The variation where the objective is
to minimize $\max_{v\in V}\cost_M(v)$ was also
considered~\cite{Low06,LeeLP11}.

The unweighted case of the semi-matching problem also received
considerable attention in the past few years. Since it was shown by
\cite{HLLT06} that an optimal solution of the semi-matching problem
is also optimal for the makespan version of the scheduling problem
(where one wants to minimize the time the last machine finishes), we
mention the results of both problems.
The problem was first studied in a special case, called {\em nested}
case where, for any two jobs, if their sets of neighbors are not
disjoint, then one of these sets contains the other set. This case was
shown to be solvable in $O(m+n\log n)$ time~\cite[p.103]{Pinedo01}.
For the general unweighted semi-matching problem,
Abraham~\cite[Section 4.3]{Abraham03} and Harvey, Ladner, Lov\'asz and
Tamir~\cite{HLLT06} independently developed two algorithms with
$O(nm)$ running time. Lin and Li~\cite{LinLi04} also gave an
$O(n^3\log{n})$-time algorithm which is later generalized to a more
general cost function~\cite{Li06}.
%
%
Recently, Lee, Leung and Pinedo~\cite{LeeLP11} showed that the problem
can be solved in polynomial time even when there are release times.

Recently after the preliminary version of this paper appeared, the unweighted semi-matching problem has been generalized to the
quasi-matching problem by Bokal, Bresar and Jerebic~\cite{BokalBJ12}.
In this problem, a function $g$ is provided and each vertex $u\in U$
is required to connect to at least $g(u)$ vertices in $v$. Therefore,
the semi-matching problem is when $g(u)=1$ for every $u\in U$.  They
also developed an algorithm for this problem which is a generalization
of the Hungarian method and used it to deal with a routing problem
in CDMA-based wireless sensor networks.

Galc\'{\i}k, Katrenic and Semanisin \cite{GalcikKS11} very recently showed a nice reduction from the unweighted semi-matching problem to a variant of the {\em maximum bounded-degree semi-matching} problem. Their approach resulted in two algorithms. The first algorithm has the same running time as ours while the second algorithm is randomized and has a running time of $O(n^\omega)$ where $\omega$ is the exponent of the best known matrix multiplication algorithm.

Motivated by the problem of assigning wireless stations (users) to
access points, the unweighted semi-matching problem is also
generalized to the problem of finding optimal semi-matching with
minimum weight where an $O(n^2m)$ time algorithm was
given~\cite{HaradaOSY07}.

Approximation algorithms and online algorithms for this problem
(both weighted and unweighted cases) and the makespan version have
also gained a lot of attention over the past few decades and have
applications ranging from scheduling in hospital to wireless
communication network. (See \cite{Scheduling_survey,vaik05} for the
recent surveys.)

\paragraph{Applications}
As motivated by Harvey et al.~\cite{HLLT06}, even in an online setting
where jobs arrive and depart over time, they may be reassigned
from one machine to another cheaply if the algorithm's running time is
significantly faster than the arrival/departure rate. (One example of
such case is the Microsoft Active Directory
system~\cite{GLLK79-load-balancing,HLLT06}.)
The problem also arose from the Video on Demand (VoD) systems where
the load of video disks needs to be balanced while data blocks from
the disks are retrieved or while serving
clients~\cite{Low02,TamirV10}.
The problem, if solved in the distributed setting, can be used to
construct a load balanced data gathering tree in sensor
networks~\cite{SSK06,MachadoT08}. The same problem also arose in
peer-to-peer systems~\cite{SuriTZ04,KothariSTZ04,SuriTZ07}.

In this paper, we also consider an ``edge cover'' version of the
problem.
In some applications such as sensor networks, there are no
jobs and machines but the sensor nodes have to be clustered and each
cluster has to pick its own head node to gather information from
other nodes in the cluster.
Motivated by this, Harada, Ono, Sadakane and
Yamashita~\cite{HaradaOSY08} introduced the
{\em balanced edge cover} problem\footnote{This problem is also
known as a {\em constant jump system} (see, e.g.,
\cite{tamir1995,Lovasz97}).} where the goal is to find an edge cover
(set of edges incident to every vertex) that minimizes the total
cost over all vertices. (The cost on each vertex is as previously
defined.) They gave an $O(nm)$ algorithm for this problem and
claimed that it could be used to solve the semi-matching problem as
well. We show that this problem can be efficiently reduced to the
semi-matching problem. Thus, our algorithm (for unweighted case)
also gives a better bound on the balanced edge cover problem.

\subsection*{Our results and techniques}
We consider the semi-matching problem and give a faster algorithm
for each of the weighted and unweighted cases. We also extend the
algorithm for the unweighted case to solve the balanced edge cover
problem. 

\squishlist
\item \textbf{Weighted Semi-Matching:} (Section~\ref{sec:weighted})
We present an $O(nm\log{n})$ algorithm, improving the previous
$O(n^3)$ algorithm by Horn~\cite{horn1973} and Bruno et
al.~\cite{BrunoCS74}.
%
%
As in the previous results \cite{horn1973,BCS74,Harvey_slide}, we
use the reduction of the weighted semi-matching problem to the
weighted bipartite matching problem as a starting point. We, however,
only use the structural properties arising from the reduction and do
not actually perform the reduction.

\item \textbf{Unweighted Semi-Matching:} (Section~\ref{sec:unweighted})
We give an $O(\sqrt{n} m\log n)$ algorithm, improving the previous
$O(nm)$ algorithms by Abraham~\cite{Abraham03} and Harvey et
al.~\cite{HLLT06}.\footnote{We also observe an $O(n^{5/2}\log n)$
algorithm that arises directly from the reduction by applying
\cite{KLST01}.}
Our algorithm uses the same reduction to the min-cost flow problem
as in~\cite{HLLT06}. However, instead of canceling one negative
cycle in each iteration, our algorithm exploits the structure of the
graphs and the cost functions to cancel many negative cycles in a
single iteration. This technique can also be generalized to any convex
cost function.

\item \textbf{Balanced Edge Cover:} (Section~\ref{sec:edge-cover})
We also present a reduction from the balanced edge cover problem to
the unweighted semi-matching problem. This leads to an
$O(\sqrt{n}m\log n)$ algorithm for the problem, improving the
previous $O(nm)$ algorithm by Harada et al.~\cite{HaradaOSY08}.
The main idea is to identify the ``center'' vertices of all the
clusters in the optimal solution. (Note that any balanced edge cover
(in fact, any minimal edge cover) clusters the vertices into stars.)
Then, we partition the vertices into two sides, center and
non-center ones, and apply the semi-matching algorithm on this
graph.

\squishend

\section{Weighted semi-matching}
\label{sec:weighted}

In this section, we present an algorithm that finds an optimal weighted
semi-matching in $O(nm\log n)$ time.

\subsection*{Overview}

Our improvement follows from studying the reduction from the
weighted semi-matching problem to the weighted bipartite matching
problem considered in the previous
works~\cite{horn1973,BrunoCS74,Harvey_slide} and the
Edmonds-Karp-Tomizawa (EKT) algorithm for finding the weighted
bipartite matching~\cite{EK72,Tomizawa71}.
We first review these briefly.
For more detail, see Appendix~\ref{sec:EK_algo}
and~\ref{sec:bimatching-algo}.

\paragraph{Reduction} As in~\cite{horn1973,BrunoCS74,Harvey_slide}, we
consider the reduction from the semi-matching problem on a bipartite
graph $G=(U\cup V,E)$ to the minimum-weight bipartite matching on a
graph $\hat G$.
The reduction is done by {\em exploding} the vertices in $V$, i.e., for each
vertex $v\in V$, we create $\deg(v)$ vertices, $v^1, v^2, \ldots,
v^{\deg(v)}$. We also make copies of edges incident to $v$ in the
original graph $G$, i.e, for each vertex $u\in U$ such that $uv\in E$,
we create edges $uv^1, uv^2, \ldots, uv^{\deg(v)}$. For each edge
$uv^i$ incident to $v^i$ in $\hat G$, we set its weight to $i$ times
its original weight in $G$, i.e, $w_{uv^i}=i\cdot w_{uv}$. We denote
the set of these vertices by $\hat V_v$.
Thus, we have
\begin{align*}
\hat{G} &= (U\cup\hat{V},\hat{E}) \\
\hat{V} &= \{v^1,v^2,\ldots,v^{\deg_G(v)}:v\in V\} \\
\hat{E} &= \{uv^1,uv^2,\ldots,v^{\deg_G(v)}:uv\in E\} \\
\hat{w}_{uv^i} &= i\cdot w_{uv}\quad \forall uv\in
E, i\in\{1,2,\ldots,\deg_G(v)\}
\end{align*}

The correctness of this reduction can be seen by
replacing the edges incident to $v$ in the semi-matching by the
edges incident to $v^1, v^2, \ldots$ with weights in decreasing
order.  For example, in Figure~\ref{subfig:reduction}, edge $u_1v_1$
and edge $u_2v_1$ in the semi-matching in $G$ correspond to
$u_1v_1^1$ and $u_2v_1^2$ in the matching in $\hat G$.
The reduction is illustrated in Figure~\ref{subfig:reduction}.

This alone does not give an improvement on the semi-matching problem
because the number of edges becomes $O(nm)$.
However, we can apply some tricks to improve the running time.

\begin{figure}
\centering
\subfigure[Reduction] {
\includegraphics[width=0.45\textwidth]{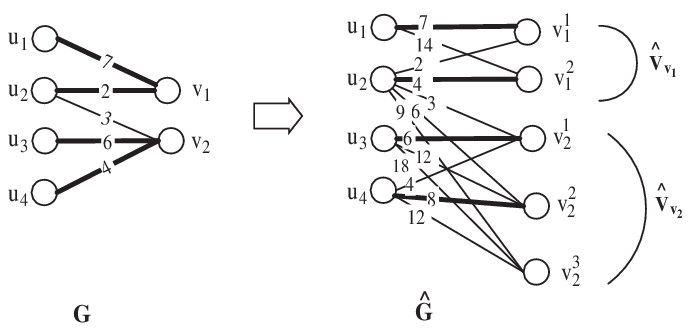}
\label{subfig:reduction} }
\subfigure[Residual graphs] {
\includegraphics[width=0.45\textwidth]{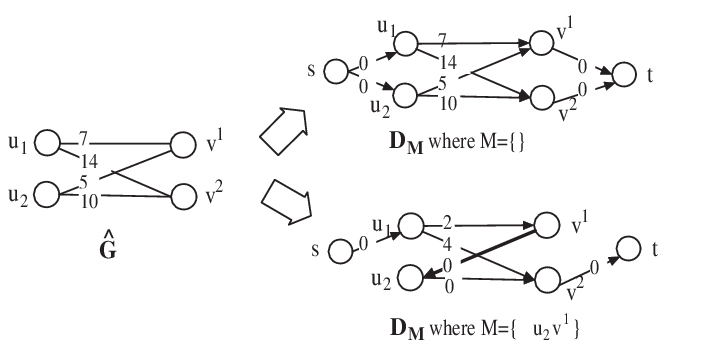}
\label{subfig:residual-graph} }
\caption{}
\label{fig:weight-reduction-soda}
\end{figure}

\paragraph{EKT algorithm} Our improvement comes from studying the behavior
of the EKT algorithm for finding the bipartite matching in $\hat G$.
%
The EKT algorithm iteratively increases the cardinality of the
matching by one by finding a shortest augmenting path.  Such path
can be found by applying Dijkstra's algorithm on the {\em residual
graph} $D_M$ (corresponding to a matching $M$) with a {\em reduced
cost}, denoted by $\tilde w$ as an edge length.\\

Figure~\ref{subfig:residual-graph} shows examples of a residual graph
$D_M$. The direction of an edge depends on whether it is in the
matching or not. The weight of each edge depends on its weight in the
original graph and the costs on its end-vertices. We draw an edge of
length 0 from $s$ to all vertices in $U_M$ and from all vertices in
$\hat V_M$ to $t$, where $U_M$ and $\hat V_M$ are the sets of
unmatched vertices in $U$ and $\hat V$, respectively. We want to find
the shortest path from $s$ to $t$ or, equivalently, from $U_M$ to
$\hat V_M$.

The reduced cost is computed from the {\em potentials} on the
vertices, which can be found as in
Algorithm~\ref{algo:EKT}.\footnote{Note that we set the
potentials in an unusual way: We keep potentials of the unmatched
vertices in $\hat V$ to $0$. The reason is roughly that we can speed
up the process of finding the distances of all vertices but vertices
in $\hat V_M$. Notice that this type of potentials is valid too
(i.e., $\tilde w$ is non-negative) since for any edge $uv$ such that
$v\in \hat V_M$ is unmatched,
$\tilde{w}_{uv}=w_{uv}+p(u)-p(v)=w_{uv}+p(u)\geq 0$.}

\begin{algorithm}
\caption{\pname{EKT Algorithm} $(\hat G, w)$} \label{algo:EKT}
\begin{algorithmic}[1]
  \STATE Let $M=\emptyset$.
  \STATE For every node $v$, let $p(v)=0$. ($p(v)$ is a potential on
  $v$.)
  \REPEAT{
  \STATE\label{line:begin_iteration} Let $\tilde w_{uv}=w_{uv}+p(u)-p(v)$ for every edge
  $uv$. ($\tilde w_{uv}$ is a reduced cost of an edge $uv$.)
  \STATE For every node $v$, compute the distance $d(v)$ which is
  the distance from $U_M$  (the set of unmatched vertices in $U$)
  to $v$ in $D_M$. (Recall that the length of edges in $D_M$ is
  $\tilde w$.)
  \STATE Let $P$ be the shortest $U_M$-$\hat V_M$ path in $D_M$.
  \STATE Update the potential $p(u)$ to $d(u)$ for every
  vertex $u\in U\cup (\hat V\setminus \hat V_M)$.
  \STATE Augment $M$ along $P$, i.e., $M=
  P\triangle M$ (where $\triangle$ denotes the symmetric difference operator).
  }
  \UNTIL {all vertices in $U$ are matched}
  \RETURN $M$
\end{algorithmic}
\end{algorithm}

Applying EKT algorithm directly leads to an $O(n(n'\log n' +
m'))$-time algorithm
where $n=|U|$, $n'=|U\cup\hat{V}|$ and $m'=|\hat{E}|$ are the number of
vertices and edges in $\hat G$.
Since $|\hat V|=\Theta(m)$ and $m'=O(nm)$, the running time is
$O(nm\log n+n^2m)$.
(We note that this could be brought down to $O(n^3)$ by applying the
result of Kao, Lam, Sung and Ting~\cite{KLST01} to reduce the number
of participating edges. See Appendix~\ref{sec:bimatching-algo}.)
The bottleneck here is the Dijkstra's algorithm which needs $O(n'\log
n'+m')$ time. We now review this algorithm and pinpoint the part that
will be sped up.

\paragraph{Dijkstra's algorithm} Recall that the Dijkstra's algorithm
starts from a source vertex and keeps adding to its shortest path
tree a vertex with minimum tentative distance. When a new vertex $v$
is added, the algorithm updates the tentative distance of all
vertices outside the tree by relaxing {\em all} edges incident to
$v$. On an $n'$-vertex $m'$-edge graph, it takes $O(\log n')$ time
(using priority queue) to find a new vertex to add to the tree and
hence $O(n'\log n')$ in total. Further, relaxing all edges takes
$O(m')$ time in total.
Recall that in our case, $m'=O(nm)$ which is too large. {\em
Thus, we wish to reduce the number of edge relaxations to improve
the overall running time.}

\paragraph{Our approach} We reduce the number of edge relaxation as
follows. Suppose that a vertex $u\in U$ is added to the shortest
path tree. For every $v\in V$, a neighbor of $u$ in $G$, we relax
all edges $uv^1$, $uv^2$, $\ldots$, $uv^i$ in $\hat G$ {\em at the
same time}. In other words, instead of relaxing $O(nm)$ edges
in $\hat G$ separately, we group the edges to $m$ groups (according
to the edges in $G$) and relax all edges in each group together. We
develop a relaxation method that takes $O(\log n)$ time per group.
In particular, we design a data structure $H_v$, for each vertex
$v\in V$, that supports the following operations.

\squishlist
\item {\sc Relax}($uv$, $H_v$): This operation works as if it relaxes
edges $uv^1$, $uv^2$, $\ldots$
\item {\sc AccessMin}($H_v$): This operation returns a vertex $v^i$ (exploded from
  $v$) with minimum tentative distance among vertices that are not deleted (by the next
  operation).
\item {\sc DeleteMin}($H_v$): This operation finds $v^i$ from {\sc
AccessMin} and then returns and deletes $v^i.$ %
\squishend

Our main result is that, by exploiting the structure of the problem,
one can design $H_v$ that supports {\sc Relax}, {\sc AccessMin} and
{\sc DeleteMin} in $O(\log n)$, $O(1)$ and $O(\log n)$ respectively.
Before showing such result, we note that speeding up Dijkstra's
algorithm and hence EKT algorithm is quite straightforward once we
have $H_v$: We simply build a binary heap $H$ whose nodes correspond
to vertices in an original graph $G$. For each vertex $u\in U$, $H$
keeps track of its tentative distance. For each vertex $v\in V$, $H$
keeps track of its {\em minimum tentative distance} returned from
$H_v$.

\paragraph{Main idea in designing $H_v$} Before going into details, we sketch the main idea
here. The data structure $H_v$ that allows fast ``group relaxation''
operation can be built because of the following nice structure of
the reduction: For each edge $uv$ of weight $w_{uv}$ in $G$, the
weights $w_{uv^1}, w_{uv^2}, \ldots$ of the corresponding edges in
$\hat G$ increase linearly (i.e., $w_{uv}, 2w_{uv}, 3w_{uv},
\ldots$). This enables us to know the order of vertices, among $v^1,
v^2, \ldots$, that will be added to the shortest path tree.  For
example, in Figure~\ref{subfig:residual-graph}, when $M=\emptyset$,
we know that, among $v^1$ and $v^2$, $v^1$ will be added to the
shortest path tree first as it always has a smaller tentative
distance.

However, since the length of edges in $D_M$ does not solely depend
on the weights of the edges in $\hat G$ (in particular, it also
depends on potentials on both end-vertices), it is possible (after
some iterations of the EKT algorithm) that $v^1$ is added to the
shortest path tree after $v^2$.

Fortunately, due to the way the potential is defined by the EKT
algorithm, a similar nice property still holds: Among $v^1, v^2,
\ldots$ in $D_M$ corresponding to $v$ in $G$, if a vertex $v^k$, for
some $k$, is added to the shortest path tree first, then the
vertices on each side of $v^k$ have a nice order: Among $v^1, v^2,
\ldots, v^{k-1}$, the order of vertices added to the shortest path
tree is $v^{k-1}, v^{k-2}, \ldots, v^2, v^1$.  Further, among
$v^{k+1}, v^{k+2}, \ldots$, the order of vertices added to the
shortest path tree is $v^{k+1}, v^{k+2}, \ldots$.

This main property, along with a few other observations, allow us to
construct the data structure $H_v$.
In the next section, we show the properties we need and use them to
construct $H_v$ in the latter section.

\subsection{Properties of the tentative distance}
\label{sec:properties} Consider any iteration of the EKT algorithm
(with a potential function $p$ and a matching $M$). We study the
following functions $f_{*_v}$ and $g_{*v}$.

\begin{definition}
\label{def:fandg} For any edge $uv$ from $U$ to $V$ and any integer
$1\leq i\leq \deg(v)$, let \[g_{uv}(i) = d(u)+p(u)+i\cdot w_{uv}
\quad \text{and}\quad f_{uv}(i)=g_{uv}(i)-p(v^i) =
d(u)+p(u)-p(v^i)+i\cdot w_{uv}.\] %
For any $v\in V$ and $i\in [\deg(v)]$, define the {\em lower envelope}
of $f_{uv}$ and $g_{uv}$ over all $u\in U$ as
\[
f_{*v}(i)=\min_{u:uv\in E}f_{uv}(i)\quad \quad \text{and}\quad
g_{*v}(i)=\min_{u:uv\in E}g_{uv}(i).\]
\end{definition}

Our goal is to understand the structure of the function $f_{*v}$
whose values $f_{*v}(1), f_{*v}(2), \ldots$ are tentative distances
of $v^1, v^2, \ldots$, respectively. The function $g_{*v}$ is simply
$f_{*v}$ with the potential of $v$ ignored. We define $g_{*v}$ as it
is easier to keep track of since it is a combination of linear
functions $g_{uv}$ and therefore piecewise linear.
Now we state the key properties that enable us to keep track of
$f_{*v}$ efficiently. Recall that $v^1, v^2, \ldots$ are the
exploded vertices of $v$ (from the reduction).

\begin{proposition} \label{prop:main-properties}
Consider a matching $M$ and a potential $p$ at any iteration of the
EKT algorithm.

\squishlist
\item[(1)] For any vertex $v\in V$, there exists $\alpha_v$ such that
  $v^1, \ldots, v^{\alpha_v}$ are all matched and $v^{\alpha_v+1}, \ldots,
  v^{\deg(v)}$ are all unmatched.
\item[(2)] For any vertex $v\in V$, $g_{*v}$ is a piecewise linear
  function.
\item[(3)] For any $i$ and any edge $uv\in E$ where $u\in U$ and $v\in V$,
  $f_{uv}(i)=f_{*v}(i)$ if and only if $g_{uv}(i)=g_{*v}(i)$.

\item[(4)] For any edge $uv\in E$ where $u\in U$ and $v\in V$,
let $\alpha_v$ be as in (1). There exists an integer $1\leq
\gamma_{uv}\leq k$ such that for $i=1,2,\ldots,\gamma_{uv}-1$,
$f_{uv}(i)\geq f_{uv}(i+1)$ and for
$i=\gamma_{uv},\gamma_{uv}+1,\ldots,\alpha_v-1$, $f_{uv}(i)\leq
f_{uv}(i+1)$. In other words,
$f_{uv}(1),f_{uv}(2),\ldots,f_{uv}(\alpha_v)$ is a unimodal
sequence. \squishend
\end{proposition}

Figure~\ref{subfig:potential-and-g} and \ref{subfig:unimodal} show
the structure of $g_{*v}$ and $f_{*v}$ according to statement (2)
and (4) in the above proposition. By statement (3), the two pictures
can be combined as in Figure~\ref{subfig:f-and-g}: $g_{*v}$
indicates $u$ that makes both $g_{*v}$ and $f_{*v}$ minimum in each
interval and one can find $i$ that minimizes $f_{*v}$ in each
interval by looking at $\alpha_v$ (or near $\alpha_v$ in some case).

\begin{figure}
  \centering
\subfigure[$g_{*v}$ and potential function.]
  {\label{subfig:potential-and-g}
  \includegraphics[height=0.22\textwidth, clip=true, trim= 5.3cm
  9.5cm 8.9cm 1.5cm]{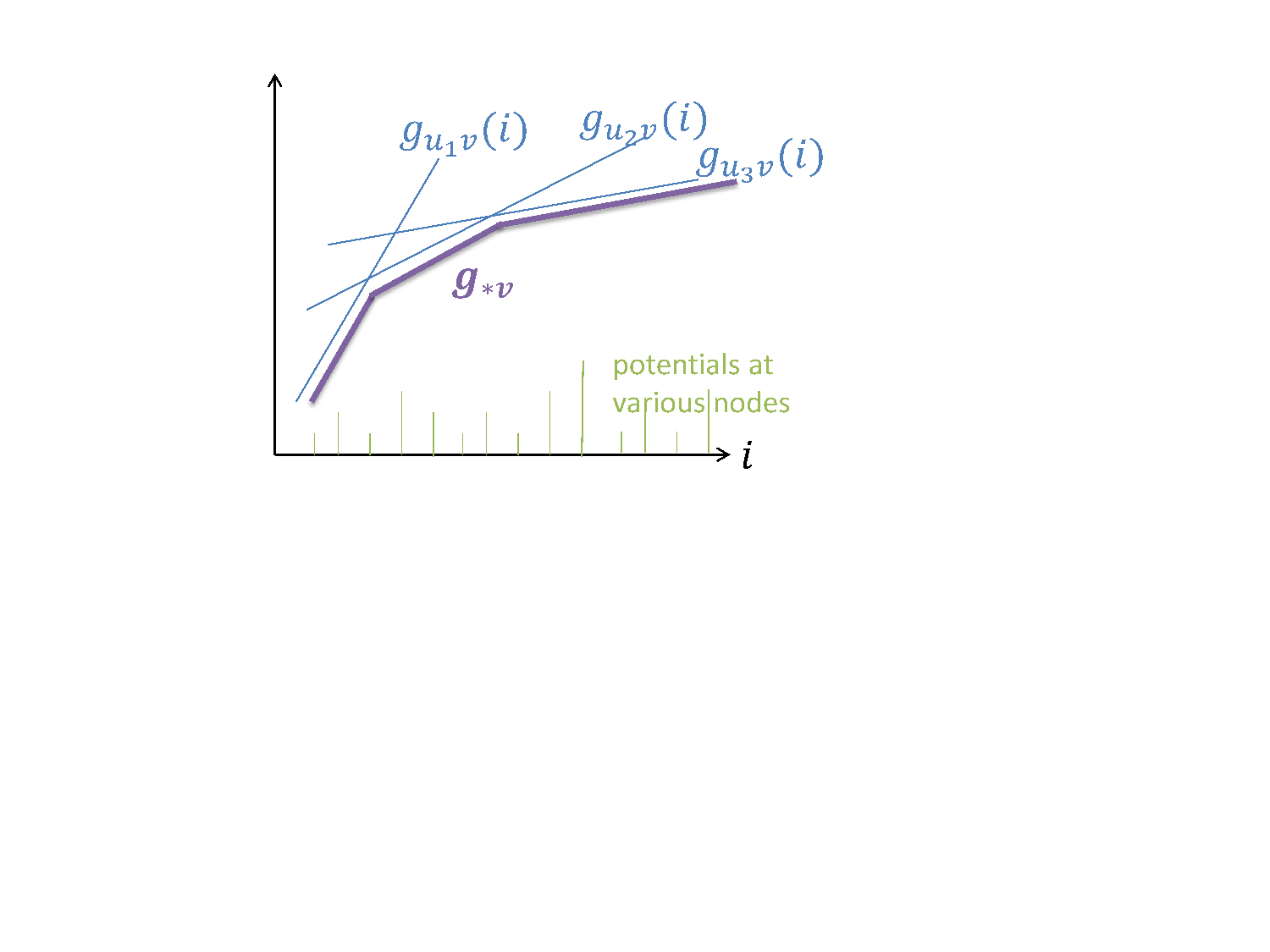}}
\subfigure[$f_{uv}$ is unimodal.]{
  \label{subfig:unimodal}
  \includegraphics[height=0.22\textwidth, clip=true, trim= 5.3cm 9.5cm 10cm 1.5cm]{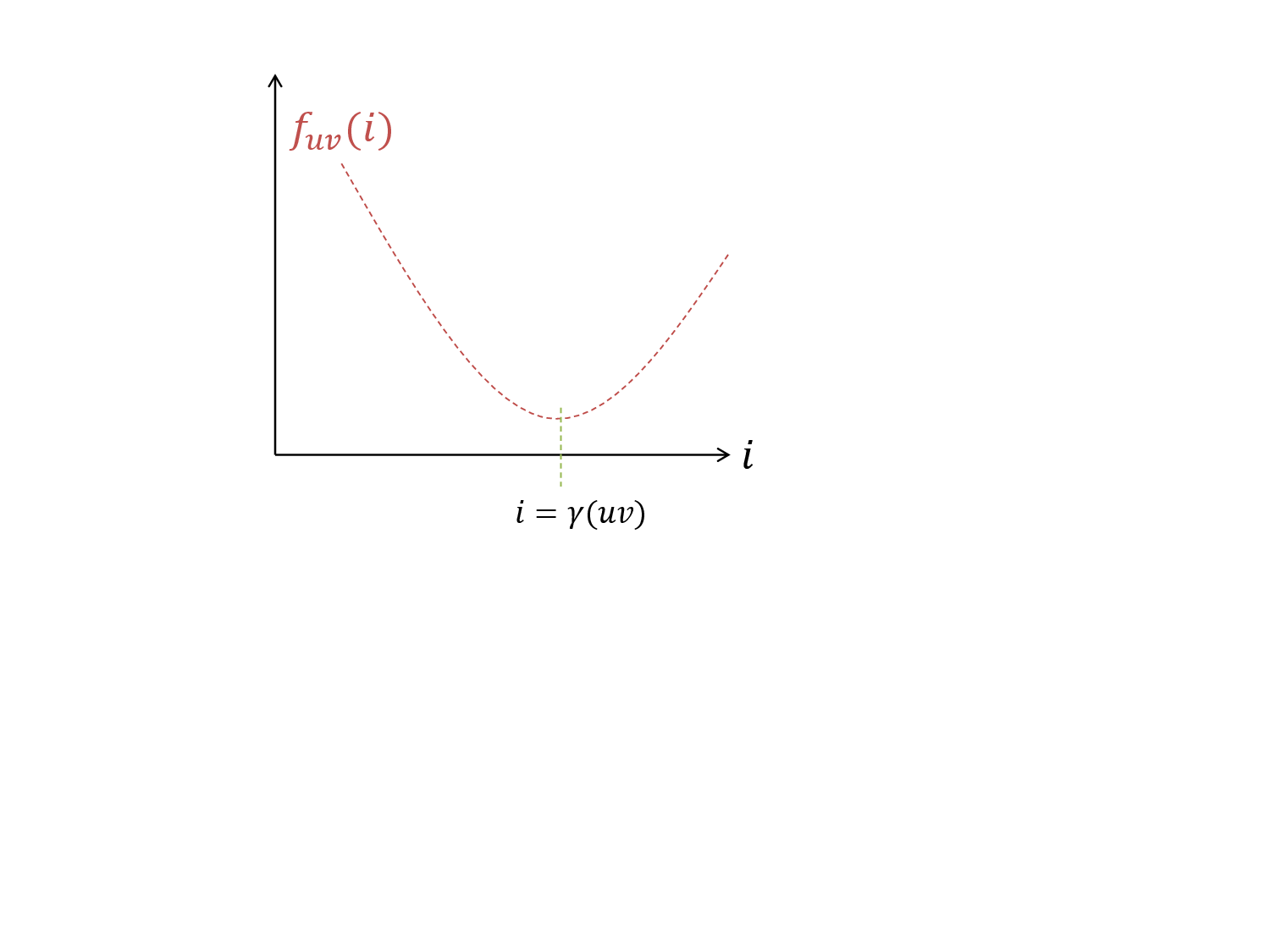}}
\subfigure[$f_{*v}$ together with $g_{*v}$.]{
  \label{subfig:f-and-g}
  \includegraphics[height=0.22\textwidth, clip=true, trim= 3.5cm 9.5cm 8.4cm 1.5cm]{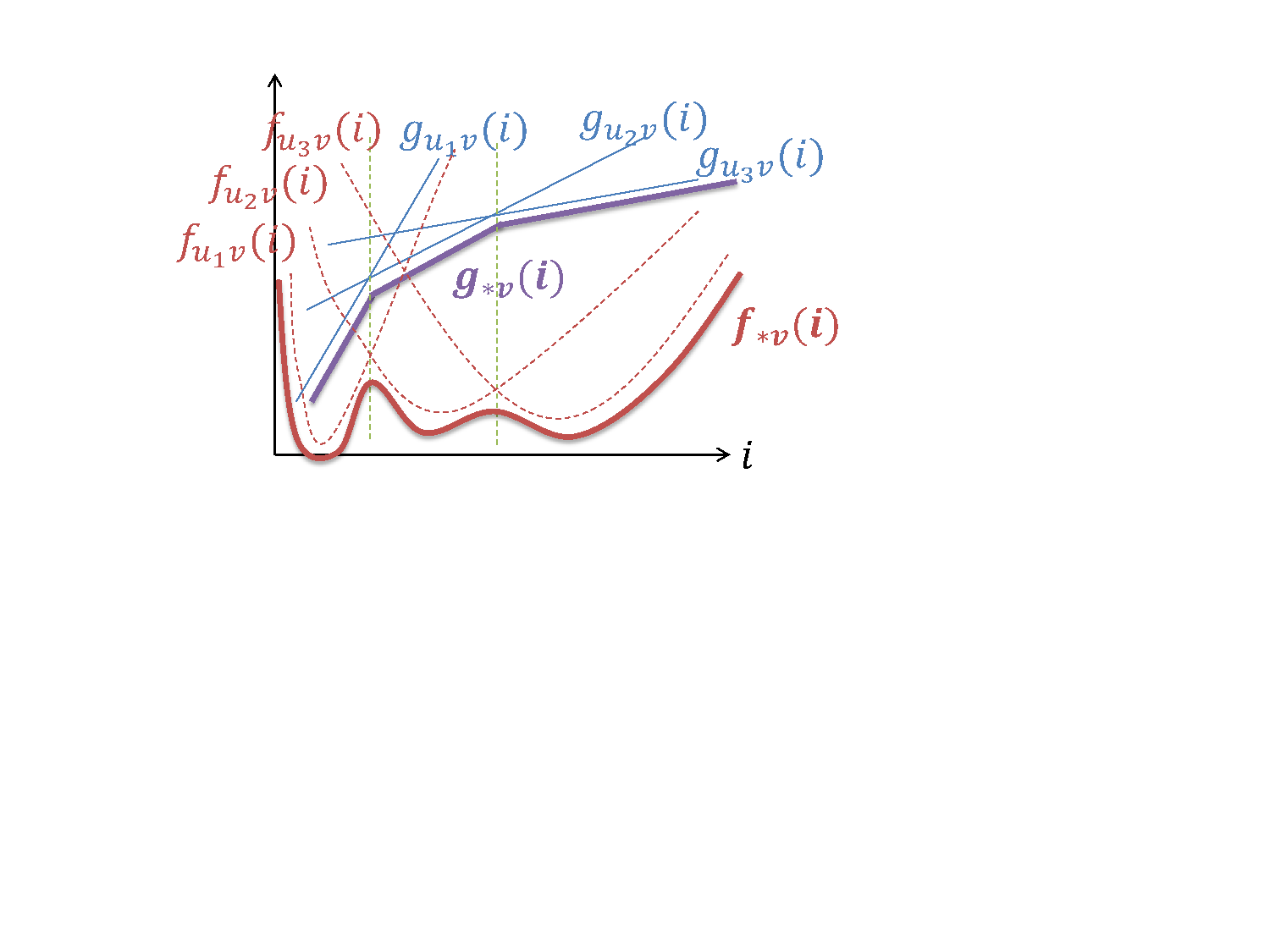}}

  \caption{Figures show graphs of potentials, $g_{*v}$, $f_{uv}$ and $f_{*v}$, where $w_{u_1v}>w_{u_2v}>w_{u_3v}>\ldots$. These functions only have value at integer points. For the sake of presentation, these functions are plotted as lines.}
  \label{fig:bowl}
\end{figure}


\begin{proof}\hfill

\noindent\textbf{(1)} The first statement follows from the following
claim.

\begin{claim}
For any $i$, if the exploded vertex $v^{i+1}$ of $v$ (in $\hat V_v$)
is matched by $M$, then $v^i$ is also matched.
\end{claim}
\begin{proof}
The claim follows from the fact that EKT algorithm maintains $M$ so
that $M$ is a so-called {\em extreme matching}, i.e., $M$ has the minimum weight among matchings of the same size. Suppose that $v^{i+1}$ is matched
by $M$ (i.e., $uv^{i+1}\in M$), but $v^i$ is not matched. Then we
can remove $uv^{i+1}$ from $M$ and add $uv^i$ to $M$. The
resulting matching will have a cost less than $M$ but have the same
cardinality, a contradiction.
\end{proof}

\noindent\textbf{(2)} To see the second statement, notice that
$g_{uv}=d(u)+p(u)+i\cdot w_{uv}$ is linear for a fixed $uv\in E$.
Hence, $g_{*v}$ is a lower envelope of a linear function, implying
that it is piecewise linear.\\

\noindent\textbf{(3)} To prove the third statement, recall that for
any $u$ and any $i$, $f_{uv}(i)=g_{uv}(i)-p(v^i)$. Therefore, for
any $u$, $u'$ and $i$, $f_{uv}(i)>f_{u'v}(i)$ if and only if
$g_{uv}(i)>g_{u'v}(i)$. Thus, the third statement follows.\\

\noindent\textbf{(4)} For the fourth claim, we first explain the
intuition. First, observe that the function $g_{uv}$ is increasing
with rate $w_{uv}$. Moreover, the difference of $f_{uv}(i)$ and
$f_{uv}(j)$ is a function of the potential $p(v^i)$ and $p(v^j)$ and
the multiple of edge weight $(j-i)w_{uv}$. In fact, whether the
difference is negative or positive depends on the value of these
three parameters. We show that these parameters change monotonically
and so we have the desired property.


To prove the fourth statement formally, we first prove two claims.

For the first claim below, recall that the potential of matched
vertices, at any iteration, is defined to be the distance on the
residual graph of the previous iteration. In particular, for any
$v^i\in \hat V$, there is a vertex $u\in U$ such that $p(u)+i\cdot
w_{uv}=p(v)$. (See Algorithm~\ref{algo:EKT}.)

\begin{claim} \label{lem:price_bound}
For any integer $i< \alpha_v$, consider the exploded vertices $v^i$
and $v^{i+1}$. Let $u$ and $u'$ denote two vertices in $U$ such that
$p(u)+i\cdot w_{uv}=p(v^i)$ and $p(u')+(i+1)\cdot
w_{u'v}=p(v^{i+1})$. Then $w_{uv}\geq p(v^{i+1})-p(v^i)\geq
w_{u'v}$.
\end{claim}
\begin{proof} The first part, $w_{u'v}\geq p(v^{i+1})-p(v^i)$,
follows from $p(v^i)=p(u)+i\cdot w_{uv}$ and $p(v^{i+1})\leq
p(u)+(i+1)\cdot w_{uv}.$ The second part, $p(v^{i+1})-p(v^i)\geq
w_{u'v}$, follows from $p(v^i)\leq p(u')+i\cdot w_{u'v}$ and
$p(v^{i+1})= p(u')+(i+1)\cdot w_{u'v}$.
\end{proof}

The proof of the next claim follows directly from the definition of
$f_{uv}$ (cf. Definition~\ref{def:fandg}).
\begin{claim} \label{lem:distance_implication}
For any $i<\alpha_v$, $f_{uv}(i)> f_{uv}({i+1})$ if and only if
$p(v^{i+1})-p(v^i)>w_{uv}$ and $f_{uv}(i)< f_{uv}({i+1})$ if and
only if $p(v^{i+1})-p(v^i)<w_{uv}$.
\end{claim}

Now, the fourth statement in the Proposition follows from the
following statements: For any integer $i<\alpha_v$,
\squishlist
\item[(i)] if $f_{uv}(i)>f_{uv}({i+1})$, then
$f_{uv}(j) \geq f_{uv}(j+1)$ for any integer $j<i$, and
\item[(ii)] if $f_{uv}(i)<f_{uv}({i+1})$, then
$f_{uv}(j) \leq f_{uv}(j+1)$ for any integer $i\leq j \leq \alpha_v.$
\squishend

To prove the first statement, let $u'$ be such that $p(u')+i\cdot
w_{u'v}=p(v^i)$. If $f_{uv}(i)>f_{uv}({i+1})$, then
$$p(v^i)-p(v^{i-1})\geq w_{u' v}\geq p(v^{i+1})-p(v^i)> w_{uv}$$
where the first two inequalities follow from
Claim~\ref{lem:price_bound} and the third inequality follows from
Claim~\ref{lem:distance_implication}. It then follows from
Claim~\ref{lem:distance_implication} that $f_{uv}(i-1)>f_{uv}(i)$.
The first statement follows by repeating the argument above. The
second statement can be proved similarly. This completes the proof
of the fourth statement.
\end{proof}


\subsection{Data structure}\label{sec:datastructure}

\paragraph{Specification} Let us first redefine the problem
so that we can talk about the data structure in a more general way.
We show how to use this data structure for the semi-matching problem
in the next section.

Let $n$ and $N$ be positive integers and, for any integer $i$,
define $[i]=\{1, 2, \ldots, i\}$. We would like to maintain at most
$n$ functions $f_1, f_2, \ldots, f_n$ mapping $[N]$ to a set of
positive reals. We assume that $f_i$ is given as an {\em oracle},
i.e., we can get $f_i(x)$ by sending a query $x$ to $f_i$ in $O(1)$
time.

Let $L$ and $S$ be a subset of $[N]$ and $[n]$, respectively. (As we
will see shortly, we use $L$ to keep the numbers left undeleted in
the process and $S$ to keep the functions inserted to the data
structure.) Initially, $L=[N]$ and $S=\emptyset$. For any $x\in
[N]$, let $f^*_S(x)=\min_{f_i\in S} f_i(x)$.
We want to construct a data structure $\cal H$ that supports the
following operations.

\squishlist

\item {\bf{\sc AccessMin}($\cal H$)}: Return $x\in L$ with minimum value $f^*_S$, i.e., $x=\arg\min_{x\in L}
f^*_S(x)$.

\item {\bf{\sc Insert}($f_i$, $\cal H$)}: Insert $f_i$ to
$S$.

\item {\bf {\sc DeleteMin}($\cal H$)}: Delete $x$ from $L$ where $x$
is returned from {\sc AccessMin($\cal H$)}.

\squishend

{\bf Properties:}
We assume that $f_1, f_2, \ldots$ have the following properties.

\squishlist
\item For all $i$, $f_i$ is {\em unimodal}, i.e., there is some
$\gamma_i\in [N]$ such that $f_i(1)\geq f_i(2)\geq \ldots \geq
f_i(\gamma_i) \leq f_i(\gamma_i+1)\leq f_i(\gamma_i+2)\leq \ldots
\leq f_i(N)\,.$
We assume that $\gamma_i$ is given along with $f_i$.

\item We also assume that each $f_i$ comes along with a
linear function $g_i$ where, for any $x\in [N]$, $g_i(x)=x\cdot
w_i+d_i$, for some $w_i$ and $d_i$. These linear functions have a
property that $f_i(x)=f^*_S(x)$ if and only if $g_i(x)=g^*_S(x)$,
where $g^*_S(x)=\min_{i\in S} g_i(x)$.

\item Finally, we assume that once $x$ is deleted from $L$, $f^*_S(x)$ will never change,
even after we add more functions to $S$.

\squishend
%
%
For simplicity, we also assume that $w_i\neq w_j$ for all $i\neq j$.
This assumption can be removed by taking care of the case of equal
weight in the insert operation.
We now show that there is a data structure such that every operation
can be done in $O(\log n)$ time.

\paragraph{Data structure design} We have two data structures to
maintain the information of $f_i$'s and $g_i$'s.
First, we create a data structure $T_g$ to maintain an ordered
sequence $g_{i_1}, g_{i_2}, \ldots$ such that $w_{i_1}\geq w_{i_2}\geq
\ldots$. We want to be able to insert a new function $g_i$ to $T_g$ in
$O(\log n)$ time. Moreover, for any $w$, we want to be able to find
$w_{i_{j}}$ and $w_{i_{j+1}}$ such that $w_{i_j}\leq w< w_{i_{j+1}}$
in $O(\log n)$ time. Such $T_g$ can be implemented by a balanced
binary search tree, e.g., an AVL tree.

Observe that the linear functions $g_{i_1}, g_{i_2}, \ldots$ appear in
the lower envelope in order, i.e., if $g_{i_j}(x)\geq
g_{i_{j+1}}(x)$, then $g_{i_j}(y)\geq g_{i_{j+1}}(y)$ for any $y>x$.
Therefore, we can use data structure $T_g$ to maintain the range of
values such that each $g_i$ (and therefore $f_i$) is in the lower
envelope. That is, we use $T_g$ to maintain $x_1\leq y_1\leq x_2\leq
y_2 \leq \ldots$ such that $g_i(x)=g^*_S(x)$ for all $i$ and
$x_i\leq x\leq y_i$).

Consider the value $\min_{x\in\{x_i, x_i+1, \ldots, y_i\}\cap L} f_i(x)$.
Since $f_i$ is unimodal, the minimum value of $f_i(x)$ over
$\{x_i, x_i+1, \ldots, y_i\}\cap L$ attains at the point closest
to $\gamma_i$ either from the left or from the right.
Thus, we can use two pointers $p_i$ and $q_i$ such that $x_i\leq
p_i\leq \gamma_i\leq q_i\leq y_i$ to maintain the minimum value of
$f_i$ from the left and right of $\gamma_i$, i.e.,
the minimum value $\min_{x\in\{x_i, x_i+1, \ldots, y_i\}\cap L} f_i(x)$
is either $f_i(p_i)$ or $f_i(q_i)$.
Finally, we use a binary heap $B$ to store the values $f_1(p_1),
f_2(p_2), \ldots$ and $f_1(q_1), f_2(q_2), \ldots$ so that we can search
and delete the minimum among these values in $O(\log n)$ time.

More details of the implementation of each operation are the
followings.

\squishlist

\item \textbf{{\sc AccessMin}($\mathcal{H}$)}: This operation is done by
returning the minimum value in $B$. This value is $\min (f_1(p_1),
f_2(p_2), \ldots, f_1(q_1), f_2(q_2), \ldots) = \min_{x\in L} f^*_S(x).$

\item \textbf{{\sc Insert}($f_i$, $\mathcal{H}$)}: First, insert
$g_i$ to $T_g$ which can be done as follows. Let the current ordered
sequence be $g_{i_1}, g_{i_2}, \ldots$. In $O(\log n)$ time, we find
$g_{i_j}$ and $g_{i_{j+1}}$ such that $w_{i_j}\leq w_i<w_{i_{j+1}}$
and insert $g_i$ between them. Moreover, we update the regions for which
$g_{i_j}$, $g_i$, and $g_{i_{j+1}}$ are in the lower envelope of
$g^*_S$, i.e., we get the values $y_{i_j}, x_i, y_i, x_{i_{j+1}},
y_{i_{j+1}}$ (note that $y_{i_j}\leq x_i \leq y_i\leq
x_{i_{j+1}}\leq y_{i_{j+1}}$).

Next, we deal with the pointers $p_i$ and $q_i$: We set
$p_i=\min(\gamma_i, y_i)$ and $q_i=\max(\gamma_i, x_i)$. (The
intuition here is that we would like to set $p_i=q_i=\gamma_i$ but
it is possible that $\gamma_i<x_i$ or $\gamma_i>y_i$ which means
that $\gamma_i$ is not in the region that $g_i$ is in the lower
envelope $g^*_S$.) Finally, we also update $p_{i_j}$ and
$q_{i_{j+1}}$: $p_{i_j}=\min(p_{i_j}, x_i)$ and
$q_{i_{j+1}}=\max(q_{i_{j+1}}, y_i)$. Figure~\ref{fig:insert} shows
an effect of inserting a new function.


\begin{figure}
\begin{center}
\subfigure[Before inserting $f_3$] {
\includegraphics[width=0.4\textwidth, clip=true, trim= 4cm 5.8cm 8cm 1cm]{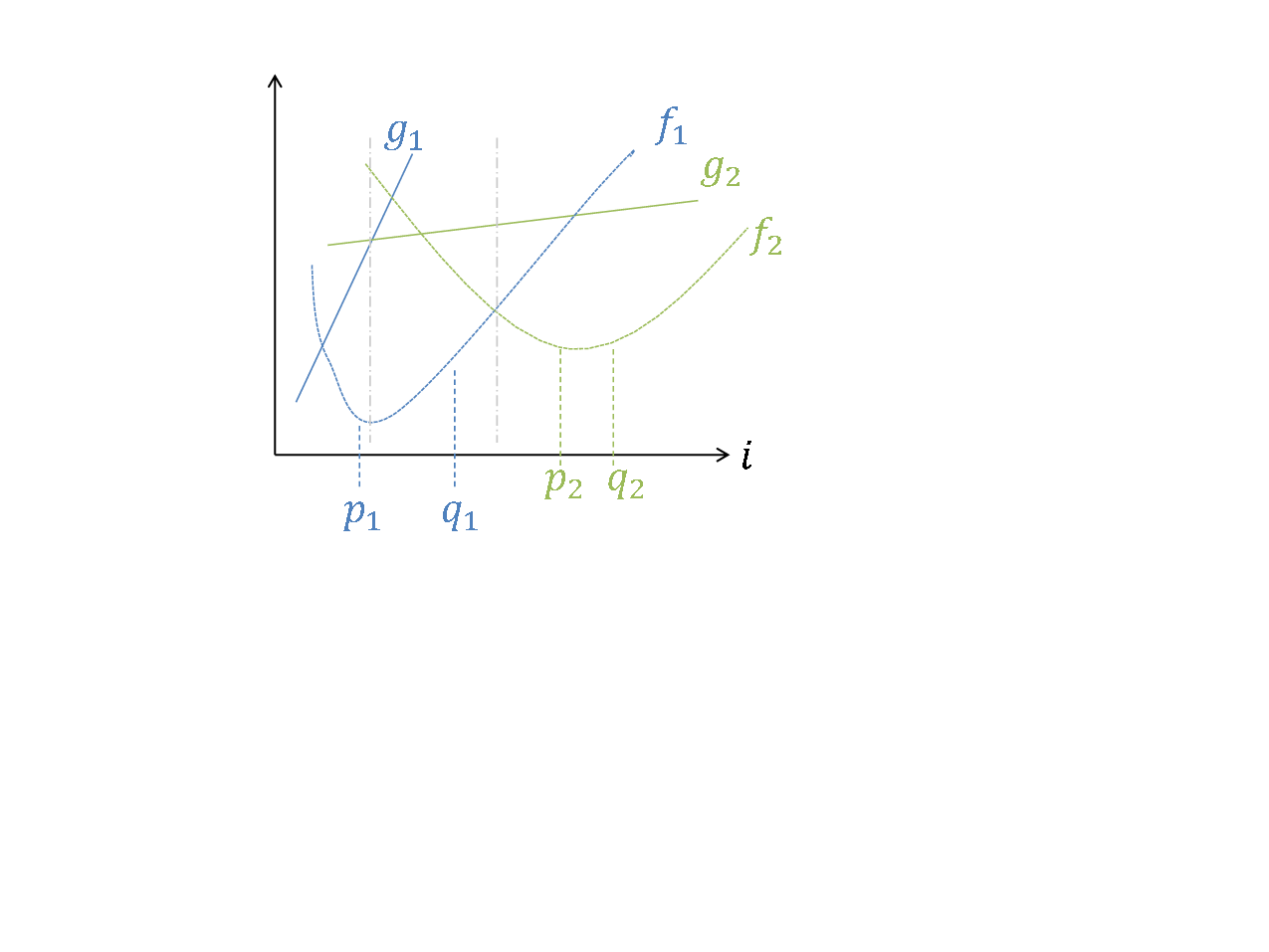}
\label{subfig:insert1} }
\subfigure[After inserting $f_3$ (value of $p_1$ is changed)] {
\includegraphics[width=0.4\textwidth, clip=true, trim= 4cm 5.8cm 8cm 1cm]{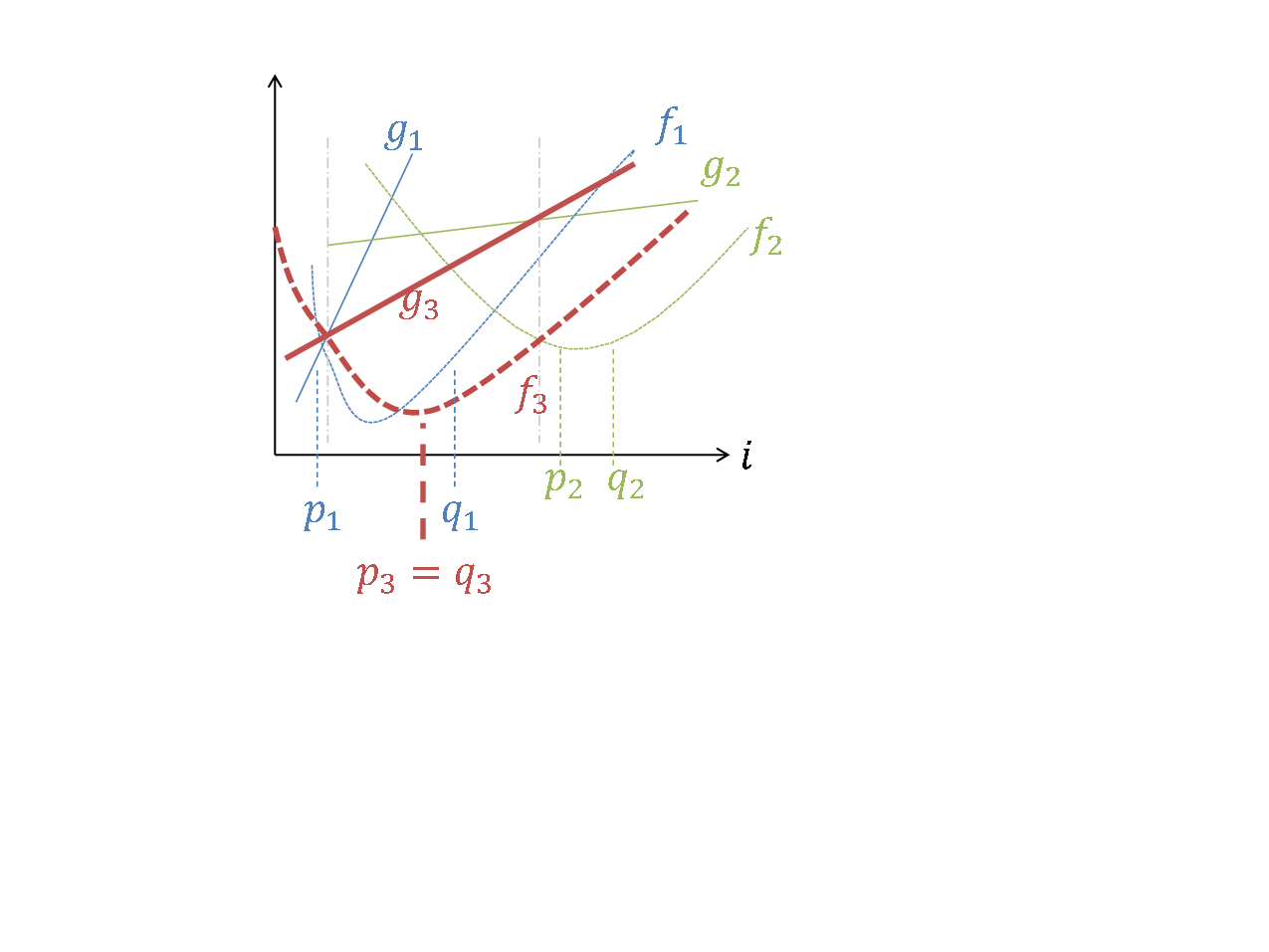}
\label{subfig:insert2} }
\caption{Inserting a new function}\label{fig:insert}
\end{center}
\end{figure}

We note one technical detail here: It is possible that $p_i$ is
already deleted from $L$. This implies that there is another
function $f_{i_{j'}}$ such that $f_{i_{j'}}(p_i)=f_i(p_i)$ (since we
assume that if $p_i$ is already deleted, then $f^*_S(p_i)$ will never
change even when we add more functions to $S$). There are two cases:
$j'<j$ or $j'>j$. For the former case, we know that
$f_{i_{j'}}(p_i-1)<f_i(p_i-1)$ since $w_{j'}>w_j$ and thus we simply
do nothing ($p_i$ will never be returned by {\sc AccessMin}). For
the latter case, we know that $f_{i_{j'}}(p_i-1)>f_i(p_i-1)$ and
thus we simply set $p_i$ to $p_i-1$.
We deal with the same case for $q_i$ similarly.

\item\textbf{{\sc DeleteMin}($\mathcal{H}$)}: We delete the node
with minimum value from $B$ (which is the one on the top of the heap).
This deleted node corresponds to one of the values $f_1(p_1),
f_2(p_2), \ldots, f_1(q_1), f_2(q_2), \ldots$. Assume that $f_i(p_i)$
(resp. $f_i(q_i)$) is such value. We insert a node with value
$f_{i}(p_i-1)$ (resp. $f_{i}(q_i+1)$).

\squishend


\subsection{Using the data structure for semi-matching
problem}\label{sec:using}

For any right vertex $v$, we construct a data structure $H_v$ as in
Section~\ref{sec:datastructure} to maintain $f_{uv}$, which comes
along with $g_{uv}$, for all neighbors of $v$. These functions
satisfy the properties above, as shown in
Section~\ref{sec:properties}. (We note that once $x$ is deleted,
$f_{*v}(x)$ will never change since this corresponds to adding a
vertex $v_x$ to the shortest path tree
with distance $f_{*v}(x)$.) 

The last issue is how to find $\gamma_{uv}$, the lowest point of an
edge $uv$ quickly.
We now show an algorithm that finds $\gamma_{uv}$, for every edge
$uv\in E$ in time $O(|V|+|E|)$ \textit{in total}. This algorithm can
be run before we start each iteration of the main algorithm (i.e.,
above Line~\ref{line:begin_iteration} of Algorithm~\ref{algo:EKT}).
To derive such algorithm, we need the following observation.
\begin{lemma}\label{lem:order_m}
Consider a vertex $v\in V$.
Let $u_1,u_2,\ldots,u_{\deg(v)}$ be vertices of $U$ incident to $v$,
where $w_{u_1v}\geq w_{u_2v} \geq\ldots\geq w_{u_{\deg(v)}v}$.
Then
$\gamma_{u_1v}\leq\gamma_{u_2v}\leq\ldots\leq\gamma_{u_{\deg(v)}v}$.
\end{lemma}
\begin{proof}
It suffices to show that if $w_{u_iv}\geq w_{u_{i+1}v}$,
then $\gamma_{u_iv}\leq\gamma_{u_{i+1}v}$.
We prove this by contrapositive.
By Claim~\ref{lem:distance_implication}, we conclude that
$\gamma_{u_iv}$ is the minimum integer $i\in [\deg(v)]$ such that
$p(v^{\gamma_{u_iv}+1})-p(v^{\gamma_{u_iv}})\leq w_{u_iv}$, and for any
$j<\gamma_{u_iv}$, $p(v^{j+1})-p(v^{j})>w_{u_iv}$.
Thus, if $\gamma_{u_iv}>\gamma_{u_{i+1}v}$, then
$w_{u_{i+1}v}\geq p(v^{\gamma_{u_{i+1}v}+1})-p(v^{\gamma_{u_{i+1}v}}) >
w_{u_{i}v}$.
This completes the proof.
\end{proof}

\paragraph{Algorithm} The following algorithm finds $\gamma_{uv}$ for all
$uv\in E$. First, in the preprocessing step (which is done once
before we begin the main algorithm), we order edges incident to $v$
decreasingly by their weights, for every vertex $v\in V$. This
process takes $O(\deg(v)\log(\deg(v)))$ time. We only have to compute
$\gamma_{uv}$ once, so this process does not affect the overall
running time.

Next, for any $v\in V$, suppose that the list is $(u_1,
u_2,\ldots,u_{\deg(v)})$. Since $w_{u_1}\geq w_{u_2}\geq\ldots\geq
w_{\deg(v)}$, it implies that $\gamma_{u_1v}\leq \gamma_{u_2v}\leq
\ldots \leq\gamma_{u_{\deg(v)}v}$ by Lemma~\ref{lem:order_m}. So, we
first find $\gamma_{u_1v}$ and then $\gamma_{u_2v}$ and so on. This
step takes $O(\deg(v))$ for each $v\in V$ and $O(m)$ in total.
Therefore, the running time for computing the minimum point
$\gamma_{uv}$'s is $O(m\log n)$.

\bigskip

We have now designed our data structure for handling the special
structure of the graph $\hat G$.
This allows us to implement the EKT algorithm on the graph $\hat G$
while the algorithm only has to read the structure of the graph $G$.
Thus, we solve the weighted semi-matching problem in $O(nm\log{n})$
time.

\section{Unweighted semi-matching}
\label{sec:unweighted}

In this section, we present an algorithm that finds the optimal
semi-matching in unweighted graph in $O(m\sqrt{n}\log n)$ time.



\subsection*{Overview}
Our algorithm consists of the following three steps.

In the first step, we reduce the problem to the min-cost flow
problem, using the same reduction from Harvey et al.~\cite{HLLT06}.
(See Figure~\ref{fig:tran-sample}.) The details are provided in
Section~\ref{subsect:min-cost-flow}. We note that the flow is
optimal if and only if there is no cost-reducing path (to be defined
later). We start with an arbitrary semi-matching and use this
reduction to get a corresponding flow. The goal is to eliminate all
the cost-reducing paths.

The second step is a divide-and-conquer algorithm used to eliminate
all the cost-reducing paths. We call this algorithm {\sc CancelAll}
(cf.  Algorithm~\ref{alg:CancelAll}).
%
%
The main idea here is to divide the graph into two subgraphs so that
eliminating cost reducing paths ``inside'' each subgraph does not
introduce any new cost reducing paths going through the other.
This dividing step needs to be done carefully.
We treat this in Section~\ref{subsec:main_unweighted_algo}.

Finally, in the last component of the algorithm we deal with
eliminating cost-reducing paths between two sets of vertices
quickly. Naively, one can do this using any unit-capacity max-flow
algorithm, but this does not give an improvement on the running time.
To get a faster algorithm, we observe that the structure
of the graph is similar to a {\em unit network}, where every vertex
has in-degree or out-degree one.  Thus, we get the same performance
guarantee as that of Dinitz's
algorithm~\cite{Dinitz70,Dinitz06}.\footnote{The algorithm is also
known as ``Dinic's algorithm''. See~\cite{Dinitz06} for details.}
Details of this part can be found in
Section~\ref{subsect:cancel-paths}.

After presenting the algorithm in the next three sections, we
analyze the running time in Section~\ref{sect:running-time}. We note
that this algorithm also works in a more general cost function
(discussed in Section~\ref{sect:general}). We also observe that there is
an $O(n^{5/2}\log n)$-time algorithm that arises directly from the
reduction of the weighted case (discussed in
Appendix~\ref{sec:bimatching-algo}). This already gives an
improvement over the previous results but our result presented here
improves the running time further.

\subsection{Reduction to min-cost flow and optimality characterization (revisited)}
\label{subsect:min-cost-flow}

In this section, we review the characterization of the optimality of
the semi-matching in the min-cost flow framework. 
%
%
%
We use the reduction as given in~\cite{HLLT06}. Given a bipartite
graph $G=(U\cup V, E)$, we construct a directed graph $N$ as
follows.  Let $\Delta$ denote the maximum degree of the vertices in
$V$.  First, add a set of vertices, called {\em cost centers},
$C=\{c_1,c_2,\ldots,c_\Delta\}$ and connect each $v\in V$ to $c_i$
with edges of capacity 1 and cost $i$, for all $1\leq i\leq
\deg(v)$. Second, add $s$ and $t$ as a source and sink vertex.  For
each vertex in $U$, add an edge from $s$ to it with zero cost and
unit capacity. For each cost center $c_i$, add an edge to $t$ with
zero cost and infinite capacity. Finally, direct each edge $e \in E$
from $U$ to $V$ with capacity 1 and cost 0. Observe that the new
graph $N$ has $O(n)$ vertices and $O(m)$ edges, and any
semi-matching in $G$ corresponds to a max flow in $N$.

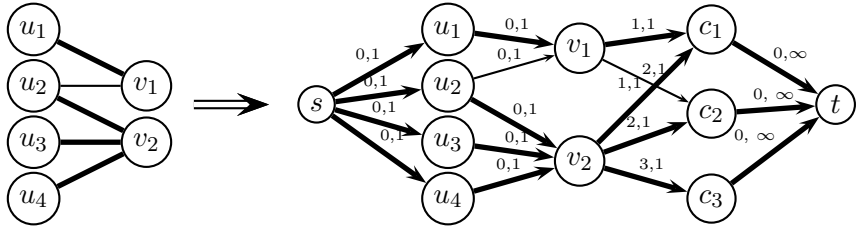
\begin{figure}
\centering
\begin{pspicture}(2, 3)
\psset{arrows=-}
\cnodeput(0,2.25){u1}{$u_1$}
\cnodeput(0,1.5){u2}{$u_2$}
\cnodeput(0,0.75){u3}{$u_3$}
\cnodeput(0,0){u4}{$u_4$}

\cnodeput(1.5,1.5){v1}{$v_1$}
\cnodeput(1.5,0.75){v2}{$v_2$}

\ncline[linewidth=2pt]{u1}{v1} \ncline{u2}{v1}
\ncline[linewidth=2pt]{u3}{v2} \ncline[linewidth=2pt]{u4}{v2}
\ncline[linewidth=2pt]{u2}{v2}
\end{pspicture}
\begin{pspicture}(1.5,3)
\psline[doubleline=true, doublesep=2pt]{->}(0,1.25)(1,1.25)
\end{pspicture}
\begin{pspicture}(8,3)
\psset{arrows=->}
\cnodeput(0,1.25){s}{$s$}
\cnodeput(1.75,2.25){u1}{$u_1$}
\cnodeput(1.75,1.50){u2}{$u_2$}
\cnodeput(1.75,0.75){u3}{$u_3$}
\cnodeput(1.75,0){u4}{$u_4$}

\cnodeput(3.5,2){v1}{$v_1$}
\cnodeput(3.5,0.5){v2}{$v_2$}

\cnodeput(5.25,2.25){c1}{$c_1$}
\cnodeput(5.25,1.13){c2}{$c_2$}
\cnodeput(5.25,0){c3}{$c_3$}

\cnodeput(6.9,1.25){t}{$t$}

\ncline[linewidth=2pt]{s}{u1} \Aput[0]{\tiny{0,1}}
\ncline[linewidth=2pt]{s}{u2} \Aput[0]{\tiny{0,1}}
\ncline[linewidth=2pt]{s}{u3} \Aput[0]{\tiny{0,1}}
\ncline[linewidth=2pt]{s}{u4} \Aput[0]{\tiny{0,1}}
\ncline[linewidth=2pt]{u1}{v1} \Aput[0]{\tiny{0,1}}
\ncline{u2}{v1} \Aput[0]{\tiny{0,1}}
\ncline[linewidth=2pt]{u3}{v2} \Aput[0]{\tiny{0,1}}
\ncline[linewidth=2pt]{u4}{v2} \Aput[0]{\tiny{0,1}}
\ncline[linewidth=2pt]{u2}{v2} \Aput[0]{\tiny{0,1}}
\ncline[linewidth=2pt]{v1}{c1} \Aput[0]{\tiny{1,1}}
\ncline{v1}{c2} \Aput[0]{\tiny{2,1}}
\ncline[linewidth=2pt]{v2}{c1} \Aput[0]{\tiny{1,1}}
\ncline[linewidth=2pt]{v2}{c2} \Aput[0]{\tiny{2,1}}
\ncline[linewidth=2pt]{v2}{c3} \Aput[0]{\tiny{3,1}}
\ncline[linewidth=2pt]{c1}{t} \Aput[0]{\tiny{0,$\infty$}}
\ncline[linewidth=2pt]{c2}{t} \Aput[0]{\tiny{0, $\infty$}}
\ncline[linewidth=2pt]{c3}{t} \Aput[0]{\tiny{0, $\infty$}}


\end{pspicture}

\medskip
\caption{Reduction to the min-cost flow problem. Each edge is
labeled with \textbf{(cost, capacity)} constraint. Thick edges
either are matching edges or contain the flow.}
\label{fig:tran-sample}
\end{figure}

Observe that the new graph $N$ contains $O(n)$ vertices and $O(m)$
edges. It can be seen that any semi-matching in $G$ corresponds to a
max flow in $N$. (See example in Figure~\ref{fig:tran-sample}.)
Moreover, Harvey et al.~\cite{HLLT06} proved that an optimal
semi-matching in $G$ corresponds to a min-cost flow in $N$; in other
words, the reduction described above is correct.
%
Our algorithm is based on observation that the largest cost is
$O(|U|)$. This allows one to use the cost-scaling framework to solve
the problem.

Now, we review an optimality characterization of the min-cost flow.
We need to define a {\em cost-reducing path} first.
Let $R_f$ denote the residual graph of $N$ with respect to a flow
$f$.
We call any path $p$ from a cost center $c_i$ to $c_j$ in $R_f$
an {\em admissible path} and call $p$ a {\em cost-reducing path} if
$i>j$. A cost-reducing path is one-to-one corresponding to a
negative cost cycle implying the condition for the minimality of
$f$.  Harvey et al.~\cite{HLLT06} proved the following.

\begin{lemma}[\cite{HLLT06}]
A flow $f$ is a min-cost flow in $N$ if and only if there is no
cost-reducing path in $R_f(N)$. \label{lmm:opt-no-aug}
\end{lemma}
 \begin{proof}
 Note that $f$ is a min-cost flow if and only if there is no negative
 cycle in $R_f$. To prove the ``only if" part, assume that there is
 an cost-reducing path from $c_i$ to $c_j$. We consider the
 shortest one, i.e., no cost center is contained in such path except the
 first and the last vertices. The edges that affect the cost of this
 path are only the first and the last ones because only edges
 incident to cost centers have cost. Cost of the first and the last
 edge is $-i$ and $j$ respectively. Connecting $c_i$ and $c_j$ with
 $t$ yields a cycle of cost $j-i<0$.

 For the ``if" part, assume that there is a negative-cost cycle in
 $R_f$. Consider the shortest cycle which contains only two cost
 centers, say $c_i$ and $c_j$ where $i>j$. This cycle contains an
 admissible path from $c_i$ to $c_j$.
 \end{proof}

Given a max-flow $f$ and a cost-reducing path $P$, one can find a
flow $f'$ with lower cost by augmenting $f$ along $P$ with a unit
flow.  This is later called {\em path canceling}.  We are now ready
to explain our algorithm.
\subsection{Divide-and-conquer algorithm}\label{subsec:main_unweighted_algo}
Our algorithm takes a bipartite graph $G=(U\cup V,E')$ and outputs
the optimal semi-matching.  It starts by transforming $G$ into a
graph $N$ as described in the previous section.  Since the source
$s$ and the sink $t$ are always clear from the context, the graph
$N$ can be seen as a tripartite graph with vertices $U\cup V\cup C$;
later on, we denote $N=(U\cup V\cup C,E)$. The algorithm proceeds by
finding an arbitrary max-flow $f$ from $s$ to $t$ in $N$ which
corresponds to a semi-matching in $G$. This can be done in linear
time since the flow is equivalent to any semi-matching in $G$.

To find the min-cost flow in $N$, the algorithm uses a subroutine
called {\pname{CancelAll}} (cf. Algorithm~\ref{alg:CancelAll}) to
cancel all cost-reducing paths in $f$. Lemma~\ref{lmm:opt-no-aug}
ensures that the final flow is optimal.  

\begin{algorithm}{
\caption{ \pname{CancelAll}$(N=(U \cup V\cup C, E))$
\label{alg:CancelAll}}
\begin{algorithmic}[1]
  \STATE \textbf{if}{$|C|=1$} \textbf{then} halt \textbf{endif}
  \STATE Divide C into $C_1$ and $C_2$ of roughly equal size.
  \STATE \pname{Cancel}($N,C_2,C_1$). \{Cancel all cost-reducing
  paths from $C_2$ to $C_1$\}.\label{line:cancel}
  \STATE Divide $N$ into $N_1$ and $N_2$ where $N_2$ is ``reachable'' from
  $C_2$ and $N_1$ is the rest.
  \STATE Recursively solve \pname{CancelAll}$(N_1)$ and \pname{CancelAll}$(N_2)$.
\end{algorithmic}
}
\end{algorithm}

{\pname{CancelAll}} works by dividing $C$ and solves the problem
recursively. Given a set of cost centers $C$, the algorithm divides
$C$ into roughly equal-size subsets $C_1$ and $C_2$ such that, for
any $c_i\in C_1$ and $c_j \in C_2$, $i<j$.  This guarantees that
there is no cost reducing path from $C_1$ to $C_2$.  Then it cancels
all cost reducing paths from $C_2$ to $C_1$ by calling {\sc Cancel}
algorithm (described in Section~\ref{subsect:cancel-paths}).

It is left to cancel the cost-reducing paths ``inside'' each of
$C_1$ and $C_2$. This is done by partitioning the vertices of $N$
(except $s$ and $t$) and forming two subgraphs $N_1$ and $N_2$.
Then solve the problem separately on each of them.
In more detail, we partition the graph $N$ by letting $N_2$ be a
subgraph induced by vertices reachable from $C_2$ in the residual
graph and $N_1$ be the subgraph induced by the remaining vertices. (Note
that both graphs have $s$ and $t$.) For example, in
Figure~\ref{fig:tran-sample}, $v_1$ is reachable from $c_3$ by the
path $c_3, v_2, u_2, v_1$ in the residual graph.

\begin{lemma}\label{lem:cancelall}
{\pname{CancelAll}}$(N)$ (cf. Algorithm~\ref{alg:CancelAll}) cancels
all cost-reducing paths in $N$.
\end{lemma}
\begin{proof}
Recall that all cost-reducing paths from $C_2$ to $C_1$ are canceled
in line~\ref{line:cancel}. Let $S$ denote the set of vertices
reachable from $C_2$.

\begin{claim}
After line~\ref{line:cancel}, no admissible paths between two cost
centers in $C_1$ intersect $S$. \label{claim:no-intersect}
\end{claim}
\begin{proof}
Assume, for the sake of contradiction, that there exists an
admissible path from $x$ to $y$, where $x,y\in C_1$, that contains a
vertex $s\in S$. Since $s$ is reachable from some vertex $z\in C_2$,
there must exist an admissible path from some vertex in $z$ to $y$;
this leads to a contradiction.
\end{proof}

This claim implies that, in our dividing step, all cost-reducing
paths between pairs of cost centers in $C_1$ remain entirely in
$N_1$. Furthermore, vertices in any cost reducing path between pairs
of cost centers in $C_2$ must be reachable from $C_2$; thus, they
must be inside $S$. Therefore, after the recursive calls, no
cost-reducing paths between pairs of cost centers in the same
subproblems $C_i$ are left.  The lemma follows if we can show that
in these processes we do not introduce more cost-reducing paths from
$C_2$ to $C_1$. To see this, note that all edges between $N_1$ and
$N_2$ remain untouched in the recursive calls.  Moreover, these
edges are directed from $N_1$ to $N_2$, because of the maximality of
$S$. Therefore there is no admissible path from $C_2$ to $C_1$.
\end{proof}

\subsection{Canceling paths from $C_2$ to $C_1$}
\label{subsect:cancel-paths} In this section, we describe an
algorithm that cancels all admissible paths from $C_2$ to $C_1$ in
$R_f$, which can be done by finding a max flow from $C_2$ to $C_1$.
To simplify the presentation, we assume that there is a super-source
$s$ and super-sink $t$ connecting to vertices in $C_2$ and in $C_1$,
respectively.


To find a maximum flow, observe that $N$ is unit-capacity and every
vertex of $U$ has indegree $1$ in $R_f$. By exploiting these
properties, we show that Dinitz's blocking flow
algorithm~\cite{Dinitz70} can find a maximum flow in
$O(|E|\sqrt{|U|})$ time. The algorithm is done by repeatedly
augmenting flows through the shortest augmenting paths (see
Appendix~\ref{sec:Dinitz}).

\begin{lemma}
Let $d_i$ be the length of the shortest $s-t$ path in the residual
graph at the $i^{th}$ iteration. For all $i$, $d_{i+1}>d_i$.
\end{lemma}

The lemma can be used to show that Dinitz's algorithm terminates
after $n$ rounds of the blocking flow step, where $n$ is the number
of vertices.  Since after the $n$-th round, the distance between the
source is more than $n$, which means that there is no augmenting
path from $s$ to $t$ in the residual graph.  The number of rounds
can be improved for certain classes of problems.  Even and
Tarjan~\cite{ET75} and Karzanov~\cite{Karzanov73} showed that in unit
capacity networks, Dinitz's algorithm terminates after
$\min(n^{2/3},m^{1/2})$ rounds, where $m$ is the number of edges.
Also, in unit networks, where every vertex has in-degree one or
out-degree one, Dinitz's algorithm terminates in $O(\sqrt{n})$ rounds
(see, e.g., Tarjan's book~\cite{TarjanBook}).
Since the graph $N$ we are considering is very similar to unit
networks, we are able to show that Dinitz's algorithm also terminates
in $O(\sqrt{n})$ in our case.


For any flow $f$, a {\em residual flow} $f'$ is a flow in a residual
graph $R_f$ of $f$.  If $f'$ is maximum in $R_f$, $f+f'$ is maximum
in the original graph.  The following lemma relates the amount of
the maximum residual flow with the shortest distance from $s$ to $t$
in our case.  The proof is a modification of Theorem~8.8
in~\cite{TarjanBook}.

\begin{lemma}\label{lem:modify-Tarjan}
If the shortest $s-t$ distance in the residual graph is $d>4$, the
amount of the maximum residual flow is at most $O(|U|/d)$.
\end{lemma}
\begin{proof}
A maximum residual flow in a unit capacity network can be decomposed
into a set $\mathcal P$ of edge-disjoint paths where the number of
paths equals the flow value.  Each of these paths are of length
at least $d$.  Clearly, each path contains the source, the sink, and
exactly two cost centers.  Now consider any path $P\in\mathcal P$ of
length $l$.  It contains $l-3$ vertices from $U\cup V$.  Since the
original graph is a bipartite graph, at least $\lfloor(l -
3)/2\rfloor\geq\lfloor(d - 3)/2\rfloor\geq(d-4)/2$ vertices are from
$U$.  Note that each path in $\mathcal P$ contains a disjoint set of
vertices in $U$, since a vertex in $U$ has in-degree one. Therefore,
we conclude that there are at most $2|U|/(d-4)$ paths in $\mathcal
P$. The lemma follows since each path has one unit of flow.
\end{proof}

From these two lemma, we have the main lemma for this section.

\begin{lemma}
{\pname{Cancel}} terminates in $O(|E|\sqrt{|U|})$
time.\label{lem:cancel_run_time}
\end{lemma}
\begin{proof}
Since each iteration can be done in $O(|E|)$ time, it is enough to
prove that the algorithm terminates in $O(\sqrt{|U|})$ rounds. The
previous lemma implies that the amount of the maximum residual flow
after the $O(\sqrt{|U|})$-th rounds is $O(\sqrt{|U|})$ units.  The
lemma thus follows because after that the algorithm augments at
least one unit of flow for each round.
\end{proof}

\subsection{Running time}
\label{sect:running-time} The running time of the algorithm is
dominated by the running time of {\pname{CancelAll}}, which can be
analyzed as follows.
Let $T(n,n',m,k)$ denote the running time of the algorithm when
$|U|=n, |V|=n', |E|=m,$ and $|C|=k$.  For simplicity, assume that
$k$ is a power of two.  By Lemma~\ref{lem:cancel_run_time},
{\pname{Cancel}} runs in $O(|E|\sqrt{|U|})$ time. Therefore,
\[
T(n,n',m,k) \leq c\cdot m\sqrt{n} + T(n_1,n'_1,m_1,k/2) +
T(n_2,n'_2,m_2,k/2),
\]
for some constant $c$, where $n_i,n'_i,$ and $m_i$ denote the number
of vertices and edges in $N_i$, respectively.  Recall that each edge
participates in at most one of the subproblems; thus, $m_1+m_2\leq
m$. Observe that the number of cost centers always decreases by a
factor of two. Thus, the recurrence is solved to
$O(\sqrt{n}m\log{k})$. Since $k=O(|U|)$, the running time is
$O(\sqrt{n}m\log{n})$ as claimed.
Furthermore, the algorithm can work with a more general cost function
with the same running time as shown in the next section.

\subsection{Generalizations of an unweighted algorithm}
\label{sect:general}

The problem can be viewed in a slightly more general version. In
Harvey et al.~\cite{HLLT06}, the cost functions for each vertex
$v\in V$ are the  same. We relax this condition, allowing a different
function for each vertex where each function is convex. More
precisely, for each $v\in V$, let $f_v:\integer_{+}\rightarrow\real$
be a convex function, i.e., for any $i$,
$f_v(i+1)-f_v(i)\geq f_v(i)-f_v(i-1)$. The cost for matching $M$ on
a vertex $v$ is $f_v(\deg_M(v))$.  In this convex cost function, the
transformation similar to what described in
Section~\ref{subsect:min-cost-flow} can still be done.
However, the number of different values of $f_v$ is now $O(|E|)$.
So, the size of the set of cost centers $C$ is now upper bounded by
$O(|E|)$ not $O(|U|)$.
Therefore, the running time of our algorithm becomes
$O(|E|\sqrt{|U|}\log{|C|})=O(|E|\sqrt{|U|}\log{|E|})
=O(\sqrt{n}m\log{n})$ (since $|E|\leq n^2$) which is the same as
before.

\section{Extension to balanced edge cover problem}
\label{sec:edge-cover}

The {\em optimal balanced edge cover} problem is defined as follows.
The input to this problem is a simple undirected graph $G =(V,E)$.
An {\em edge cover} $F\subseteq E$ is a set of edges such that every
vertex of $G$ is incident to at least one edge in $F$.
Define the cost of the edge cover $F$ as
$\cost(F)=\sum_{v\in{V}}\cost_F(v)$, where
$\cost_F(v)=\sum_{i=1}^{\deg_F(v)}i$.
(The cost function is the same as that of the unweighted semi-matching
problem\footnote{
We note that the original definition of the
balanced edge cover problem has a function $f : \integer^+
\rightarrow \real^+$ as an input~\cite{HaradaOSY08}. However, it was
shown in \cite{HaradaOSY08} that the optimal balanced edge cover can
be determined independently of function $f$ as long as $f$ is a
strictly monotonic convex function.
In other words, the problem is equivalent to the
one we define here.}.)
The goal in the optimal balanced edge cover problem is to find an edge
cover $F$ with minimum cost.

Observe that any minimal edge cover -- including any optimal balanced
edge cover -- induces a star forest; i.e., every connected component
has at most one vertex of degree greater than one (we call such
vertices centers) and the rest have degree exactly one. For any minimal edge cover $F$, we call a set of vertices $C$ an {\em extended set of centers of $F$} if (1) $C$ contains all centers of $F$, and (2) each connected component in the subgraph induced by $F$ contains exactly one vertex in $C$.
%
%

To solve the balanced edge cover problem using a semi-matching
algorithm, we first make a further observation that if we are given an extended
set of centers of an optimal balanced edge cover, then an optimal balanced edge cover can be
found by simply solving the unweighted semi-matching problem.

\begin{lemma}\label{lem:extended set of centers is enough}
Let $C$ be an extended set of centers of some optimal balanced edge
cover $F$. Let $G'=((V\setminus C)\cup C, E')$ be a bipartite graph
where $E'$ is the set of edges between $V\setminus C$ and $C$ in
$G$. Then any optimal semi-matching in $G'$ (where every vertex in
$V\setminus C$ touches exactly one edge of the semi-matching)
is an optimal balanced edge cover in $G$.
\end{lemma}
\begin{proof}
Let $M$ be any optimal semi-matching in $G'$. First, observe that $F$
is also a semi-matching in $G'$. Thus, the cost of $M$ is at most
the cost of $F$. It remains to show that $M$ is an edge cover.
In other words, we will prove that every vertex in $C$ is covered by
$M$.

Assume for the sake of contradiction that there is a vertex $v\in C$
that is not covered by $M$. We show that there exists a
cost-reducing path of $M$ starting from $v$ as follows. (The notion of
cost-reducing path is defined in Section~\ref{sec:unweighted}.)
Starting from $v_0=v$, let $v_1$ be any vertex adjacent to $v_0$ in
$F$. Such $v_1$ clearly exists since $F$ is an edge cover.
Let $v_2$ be a vertex in $C$ adjacent to $v_1$ in $M$. Such $v_2$
exists and is unique since $v_1$ has degree exactly one in $M$.
If $\deg_M(v_2)>1$, then we stop the process. Otherwise, we repeat the
process by finding a vertex $v_3$ adjacent to $v_2$ in $F$ and a
vertex $v_4$ adjacent to $v_3$ in $M$. We repeat this until we find
$v_{2k}$, for some $k$, such that $\deg_M(v_2)>1$.
This process is illustrated in Figure~\ref{fig:semi-cover}.

\begin{figure}
\centering
\includegraphics[scale=0.6]{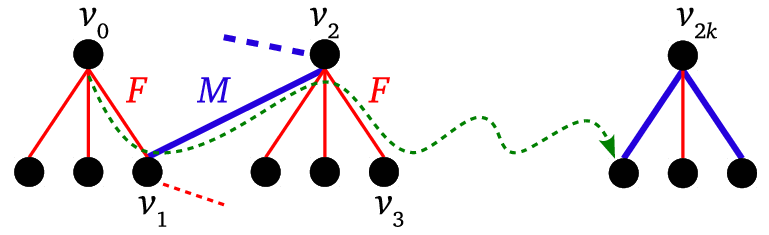}
\caption{A figure illustrates the construction of a cost-reducing path
in Lemma~\ref{lem:extended set of centers is enough}. Solid thin edges
(in red) denote edges in $F$. Solid thick edges (in blue) denote edges
in $M$. Dashed thin edges (in red) represent the fact that vertices
$v_i$, for all odd $i$, have degree exactly one in $F$. Similarly,
dashed thick edges (in blue) represent the fact that vertices $v_i$, for
all even $0<i<2k$, have degree exactly one in $M$ .
}
\label{fig:semi-cover}
\end{figure}

\begin{claim}
All vertices found during the process are distinct.
\end{claim}
\begin{proof}
Let $v_i$ be the first vertex that appears for the second time,
i.e. $v_i=v_j$ for some $j<i$ and all vertices in $\{v_0, \ldots,
v_{i-1}\}$ are distinct. Let $j<i$ be such that $v_i=v_j$.

{\bf Case 1: $i$ is odd.} This means that $v_i\notin C$. It follows
that $v_i$ has degree exactly one in $F$ (this is true for every
vertex that is not in the extended set of centers $C$ of $F$). Also
note that $(v_{j-1},v_j)$ and $(v_{i-1},v_i)$ are both in $F$. Thus,
$(v_{j-1},v_j)=(v_{i-1},v_i)$. This means that $v_{j-1}=v_{i-1}$,
contradicting the assumption that $v_i$ is the first vertex that
appears for the second time.

{\bf Case 2: $i$ is even.} This means that $v_i\in C$. It follows that
$v_i$ has degree exactly one in $M$; otherwise, the process must stop
when $v_j$ is found. As in Case 1, this fact implies that
$v_{j-1}=v_{i-1}$ since $(v_{j-1},v_j)$ and $(v_{i-1},v_i)$ are both in $M$,
contradicting the assumption that $v_i$ is the first vertex that
appears for the second time. The claim is completed.
\end{proof}

%

The above claim implies that the process will stop. Since we stop at
vertex whose degree in $M$ is more than one, the path obtained by this
process is a cost-reducing path of $M$. This contradicts the
assumption that $M$ is an optimal semi-matching.
\end{proof}

It remains to find an extended set of centers. We do this using the
following algorithm.

\paragraph{Algorithm {\sc Find-Center}} First, find a minimum
cardinality edge cover $F$. Then find {\em leveling} of vertices,
denoted by $L_F$, as follows.

First, all center vertices of $F$ (i.e., all vertices with degree more
than one in $F$) are on level $1$.
For $i=1, 2, \ldots$, we define level $i+1$ by considering at any
vertex $v$ not yet assigned to any level. We pick such vertex $v$ in
any order and consider two cases.
\begin{itemize}
\item If $i$ is odd and $v$ shares an edge in $F$ with a vertex
  on level $i$, then we add $v$ to level $i+1$.
\item If is $i$ even, then we add $v$ to level $i+1$ if
  $v$ shares an edge not in $F$ with a vertex on level $i$
  and $v$ does not share an edge in $F$ with any vertex on level
  $i+1$.
\end{itemize}

We output $C$, the set of even-level vertices, as an extended set of centers.
Note that there might be some vertices that are not assigned to any
level in $L_F$.
Figure~\ref{fig:leveling} illustrates the work of the
{\sc Find-Center} algorithm.
We first find a minimum edge cover $F$ (consisting of solid edges).
Vertices $v_1$ and $v_2$, which are the centers of the two stars in $F$,
are in the first level. The leaves of the stars
(i.e., $v_3,\ldots,v_8$) are then in the second level.
Vertices $v_9$ and $v_{10}$ are both adjacent to vertices in the second
level by edges not in $F$. Thus, any of them could be in the third
level. However, since they are adjacent in $F$, they could not
be both in the third level. If we consider $v_9$ before $v_{10}$ in
the algorithm, then $v_9$ will be in level $3$ while $v_{10}$ will be
in level $4$ as in the figure. In this case, $v_{11}$ and $v_{12}$
will not be assigned to any level. In contrast, if we consider
$v_{10}$ first, then $v_{9}$, $v_{10}$, $v_{11}$ and $v_{12}$ will be
in level $4$, $3$, $5$ and $6$, respectively.

\begin{figure}
\centering
\includegraphics[scale=0.6]{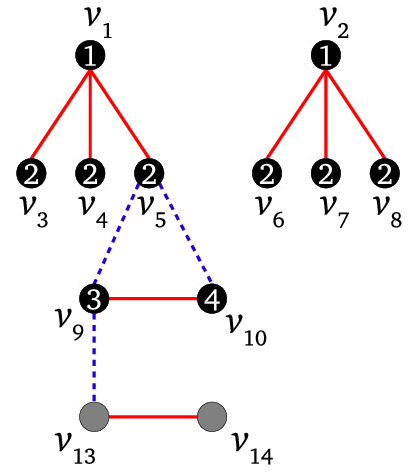}
\caption{An example of algorithm {\sc Find-Center}. The solid edges
  denote edges in the minimum cardinality edge cover $F$, and the
  dashed edges denote edges not in $F$. The numbers in vertices
  denote one possible leveling.}
\label{fig:leveling}
\end{figure}


Now we analyze the running time and show the correctness of algorithm
{\sc Find-Center}. Once we have these, the main claim of this section
follows immediately from
Lemma~\ref{lem:extended set of centers is enough}.

\paragraph{Running time analysis}
An edge cover $F$ can be constructed from a maximum cardinality matching
by adding one edge incident to each uncovered vertex
\cite{gallai1959uep,norman1959amc}.
The maximum cardinality matching in a bipartite graph can be found by
Micali-Vazirani's algorithm~\cite{MV80} in $O(\sqrt{n}m)$ time or by
Harvey's algorithm~\cite{Harvey06} in $O(n^\omega)$ time, where
$\omega$ is a time for computing matrix multiplication.
Thus, $F$ can be found in $O(\sqrt{n}m)$ time by using the first algorithm.
Moreover, finding $L_F$ could be done in a breadth-first manner,
which takes $O(n+m)$ time. Therefore, the time for the
reduction from the balanced edge cover problem to the unweighted
semi-matching problem is $O(\sqrt{n}m)$, implying the total running
time of $O(\sqrt{n}m\log n)$.

\paragraph{Correctness}
We prove the correctness by applying the algorithm BEC1 proposed in
\cite{HaradaOSY08}. This algorithm starts from any minimum edge
cover and keeps augmenting along a \textit{cost-reducing} path until
such path does not exist. Here a cost-reducing path $P$ with respect
to an edge cover $F$ is a path starting from any center vertex $u$,
following any edge in $F$ and then following an edge not in $F$.
The path $P$ keeps using edges in $F$ and edges not in $F$ alternately
until it finally uses an edge not in $F$ and ends at a vertex $v$ such that
$\deg_F(v)\leq \deg_F(u)-2$.
(See \cite{HaradaOSY08} for the formal definition.)
It was shown that BEC1 returns an optimal balanced edge cover.

\begin{lemma}
Let $C$ be a set returned from the algorithm {\sc Find-Center}.
Then $C$ is an extended set of centers of some optimal balanced edge
cover $F^*$. In other words, there exists an optimal balanced edge
cover $F^*$ such that all of its centers are in $C$, and each
connected component (in the subgraph induced by $F^*$) has exactly one
vertex in $C$.
\end{lemma}
\begin{proof}
Let $F$ be a minimum cardinality edge cover found by the algorithm
{\sc Find-Center}. Consider a variant of the algorithm BEC1 where
we augment along a shortest cost-reducing path. We will show that we
can always augment along the shortest cost-reducing path in such a way
that the parity of vertices' levels never change. To be precise, we
construct a sequence of minimum cardinality edge covers $F=F_1, F_2,
\ldots$ where we obtain $F_i$ from $F_{i-1}$ by augmenting along some
shortest cost-reducing path. By the following process, we claim
that if any vertex is on an odd (even, respectively) level in $L_F$,
then it is on an odd (even, respectively) level in
$L_{F_i}$. Moreover, if a vertex belongs to no level in $L_F$, then it
belongs to no level in $L_{F_i}$.

We prove the claim by induction on $i$.
The claim trivially holds on $F_1=F$.
Inductively, assume that the claim holds on some $F_i$.
Let $P$ be any shortest cost-reducing path with respect to $F_i$.
If there is no such path $P$, then $F^*=F_i$ is an optimal edge
cover, and we are done.
Otherwise, we consider two cases.

\begin{itemize}
  \item {\bf Case 1:} The path $P$ contains only vertices on level 1
    and 2. This is equivalent to reconnecting vertices on level 2 to
    vertices on level 1. The level of every vertex is the same in
    $L_{F_i}$ and $L_{F_{i+1}}$. Thus, the claim holds on $F_{i+1}$.
  \item {\bf Case 2:} The path $P$ contains a vertex $v_k$ not on
    level 1 or 2.
    By the construction, $v_k$ has degree one in $F$.
    Thus, $v_k$ is the end-vertex of $P$ and all other vertices are on
    level 1 and 2; otherwise, we can stop at the first vertex that is
    not on level 1 or 2 and obtain a shorter cost-reducing path.
    Specifically, we may write $P$ as $P=v_0v_1\ldots v_k$, where
    vertices $v_0,v_1\ldots,v_{k-1}$ are on level 1 and 2 alternately.
    Also, $k$ must be even since $P$ is a cost-reducing path.
    Now, let us augment from $v_0$ until we reach $v_{k-2}$.
    At this point, $v_{k-2}$ must have degree at least three (after
    the augmentation) because it is on level 1 (which means that it
    has degree more than one in $F_i$) and just got one more edge from
    the augmentation.
    If $v_k$ is on level 3, then we are done as it will be on level 1
    in $L_{F_{i+1}}$, and all vertices in its subtree will be 2 levels
    higher.
    Otherwise, $v_k$ must be on level 4. Let $a$ be a vertex on level
    3 adjacent to $v_k$ by an edge in $F_i$, which exists by the
    construction, and let $b$ be a vertex on level 2 adjacent to $a$
    by an edge not in $F_i$. There are two subcases.

    \begin{itemize}
    \item {\bf Case 2.1:} $v_{k-1}=b$. In this case, we augment along
      the path $v_1v_2\ldots v_{k-1}a$ instead.
    \item {\bf Case 2.2:} $v_{k-1}\neq b$. In this case, we get an
      edge cover with cardinality smaller than $|F_i|=|F|$ by deleting
      three edges in $F_i$ incident to vertices $b$,$v_{k-1}$,$v_k$ and
      adding edges $(a,b)$ and $(v_{k-1},v_k)$.
      (Note that for the case that $b$ is covered by an edge incident
      to $v_{k-2}$, we use the fact that $v_{k-2}$ has degree at least
      3 as discussed earlier.)
      So, this case is impossible because it contradicts the fact
      that $F$ is minimum cardinality edge cover.
    \end{itemize}
\end{itemize}
As there exist augmentations that do not change the parity of
vertices' levels, at the end of the process, we have an optimal
balanced edge cover whose extended set of centers is exactly $C$.
This completes the proof.
\end{proof}


\paragraph{Acknowledgment} We thank David Pritchard for useful
suggestions, Jane (Pu) Gao for pointing out some related surveys and
Dijun Luo for pointing out some errors in the earlier version of this
paper.

%
%
%


  \let\oldthebibliography=\thebibliography
  \let\endoldthebibliography=\endthebibliography
  \renewenvironment{thebibliography}[1]{%
    \begin{oldthebibliography}{#1}%
      \setlength{\parskip}{0ex}%
      \setlength{\itemsep}{0ex}%
  }%
  {%
    \end{oldthebibliography}%
  }

\bibliographystyle{plain}
\bibliography{semi}

\newpage
\appendix
\section*{APPENDIX}


\section{Edmonds-Karp-Tomizawa algorithm for weighted bipartite
matching} \label{sec:EK_algo}

In this section, we briefly explain Edmonds-Karp-Tomizawa (EKT)
algorithm. The algorithm starts with an empty matching $M$ and
iteratively augments (i.e., increases the size of) $M$. The matching
in each iteration is maintained so that it is {\em extreme}; i.e.,
it has the highest weight among matchings of the same cardinality. The
augmenting procedure is as follows. Let $M$ be a matching maintained
so far. Let $D_M$ be the directed graph obtained from $\hat{G}$ by
orienting each edge $e$ in $M$ from $\hat V$ to $U$ with length
$\ell_e=-w_e$ and orienting each edge $e$ not in $M$ from $U$ to
$\hat V$ with length $\ell_e=w_e$. Let $U_M$ (respectively, $\hat
V_M$) be the set of vertices in $U$ (respectively, $\hat V$) not
covered by $M$. If $|M|\neq |U|$, then there is a $U_M$-$\hat V_M$
path. Find a shortest such path, say $P$, and augment $M$ along $P$;
i.e., set $M=M\Delta P$. Repeat with the new value of $M$ until
$|M|=|U|$.

The bottleneck of this algorithm is the shortest path algorithm.
Although $D_M$ has negative-length edges, one can find a {\em
  potential} and apply Dijkstra's algorithm on $D_M$ with
non-negative {\em reduced cost}. The potential and reduced cost are
defined as follows.

\begin{definition}
A function $p:U\cup \hat{V}\rightarrow \mathbb{R}$ is a
\textit{potential} if, for every edge $uv$ in the residual graph
$D_M$, $\tilde\ell_{uv}=\ell_{uv}+p(u)-p(v)$ is non-negative. We
call $\tilde\ell$ a {\em reduced cost} with respect to a potential
$p$.
\end{definition}

The key idea of using a potential is that a shortest path from $u$
to $v$ with respect to a reduced cost $\tilde\ell$ is also a
shortest path with respect to $\ell$. We omit details here (see, e.g.,
(\cite[Chapter~7 and Section~17.2]{SchrijverA}), but note that we
can use a distance function found in the last iteration of the
algorithm as a potential, as in Algorithm~\ref{algo:EKT}.

\subsubsection*{Dijkstra's algorithm.}
\label{sec:Dijkstra} We now explain Dijkstra's algorithm on graph
$D_M$ with non-negative edge weight defined by $\tilde\ell$. Our
presentation is slightly different from the standard one but will be
easy to modify later. The algorithm keeps a subset $X$ of $U\cup
\hat{V}$, called {\em set of undiscovered vertices}, and a function
$d: U\cup\hat V\rightarrow \mathbb{R}^+$ (the \textit{tentative
  distance}). Start with $X= U\cup \hat V$ and set $d(u)=0$ for all
$u\in U_M$ and $d(v)=\infty$ for all $v\notin U_M$. Apply the
following iteratively:

\begin{algorithmic}[1]
\STATE Find $u\in X$ minimizing $d(u)$ over $u\in X$. Set
$X=X\setminus \{u\}$.
\STATE For each neighbor $v$ of $u$ in $D_M$, {\em relax} $uv$, i.e., set
$d(v)\leftarrow\min\{d(v),d(u)+\tilde\ell_{uv}\}$.
\end{algorithmic}

The running time of Dijkstra's algorithm depends on the
implementation. One implementation is by using Fibonacci heap. Each
vertex $v\in U\cup\hat V$ is kept in the heap with key $d(v)$.
Finding and extracting a vertex of minimum tentative distance can be
done in an amortized time bound of $O(\log |U\cup\hat V|)$ by
{\em extract-min} operation, and relaxing an edge can be done in an amortized time
bound of $O(1)$ by {\em decrease-key} operation.

Consider the running time of finding a shortest path.
Let $n=|U\cup V|$ and $m=|E|$.
We have to call insertion $O(n)$ times, decrease-key $O(m)$
times, and extract-min $O(n)$ times.
Thus, the overall running time is $O(m + n\log{n})$.


\section{Observation: $O(n^3)$ and $O(n^{5/2}\log (nW))$ time algorithms}
\label{sec:bimatching-algo}

We first recall the reduction from the weighted semi-matching problem
to the weighted bipartite matching problem, or equivalently, the
assignment problem. Given a bipartite graph $G=(U\cup V,E)$ with edge
weight $w$, an instance for the semi-matching problem, we construct a
bipartite graph $G=(U\cup\hat{V},\hat{E})$ with weight $\hat{w}$, an
instance for the weighted bipartite matching problem, as follows.
For every vertex $v\in V$ of degree $deg(v)$, we create {\em exploded}
vertices $v^1, v^2,\ldots, v^{deg(v)}$ in $\hat{V}$ and let
$\hat{V}_{v}$ denote a set of such vertices.
For each edge $uv$ in $E$ of weight $w_{uv}$, we also create $deg(v)$
edges $uv^1,uv^2,\ldots,uv^{deg(v_i)})$, with associated weights $w_{uv},
2\cdot w_{uv},\ldots, \deg(v)\cdot w_{uv}$, respectively.
It is easy to verify that finding optimal semi-matching in $G$ is
equivalent to finding a minimum matching in
$\hat{G}$. Figure~\ref{subfig:reduction} shows an example of this
reduction.

The construction yields a graph $\hat{G}$ with $O(m)$ vertices and
$O(nm)$ edges. Thus, applying any existing algorithm for the
weighted bipartite matching problem directly is not enough to get an
improvement.
However, we observe that the reduction can be done in
$O(n^2\log n)$ time, and we can apply the result of Kao et
al.~in~\cite{KLST01} to reduce the number of participating edges to
$O(n^3)$. Thus, Gabow and Tarjan's scaling algorithm~\cite{GT89} gives
us the following result.

\begin{observation}
If all edges have non-negative integer weight bounded by $W$, then
there is an algorithm for the weighted semi-matching problem with the
running time of $O(n^{5/2}\log (nW))$.
\label{obs:bimatching-algo}
\end{observation}

This result immediately gives an $O(n^{5/2}\log n)$ time algorithm for
the unweighted case (i.e., $W=1$). Hence, we already have an
improvement upon the previous $O(nm)$ time algorithm for the case of
dense graph.

Now, we give an explanation on the observation.
If we reduce the problem normally (as in
Section~\ref{sec:weighted}) to get $\hat G$, then the number of
edges in $\hat G$ and the running time will be $O(nm)$. However,
since the size of any matching in the graph $\hat G$ is at most
$|U|$, it suffices to consider only the smallest $|U|$ edges in
$\hat G$ incident to each vertex in $U$. Therefore, we may assume
that $\hat G$ has $O(n^2)$ edges. (The same observation is also used
in \cite{KLST01}.)

More precisely, let $E_u$ be a set of edges incident to $u$ in $\hat
G$, and $R$ be a set of $|U|$ smallest edges of $E_u$. If the
maximum matching of minimum weight, say $M$, contains an edge
$e\in E_u\setminus R$, then $R\cup\{e\}$ has $|U|+1$ edges. This implies
that there is an edge $e'\in R$ incident to a vertex $v\in\hat V$
not matched by $M$. Thus, we can replace $e$ by $e'$ which results
in a matching of smaller weight. Therefore, we need to keep only
$|U|^2$ edges in our reduction.
Moreover, we can also reduce the time of the reduction to
$O(n^2\log{n})$.

The faster reduction is applied at each vertex $u\in U$ as follows.
First, we create a binary graph $H$.
Each node of $H$ has a key $(e=uv,i)$ and a value $i\cdot w_e$, where
$e=uv\in E$ and $i$ is an integer.
In other words, the value of the node in $H$ with key $(e,i)$ is the weight of an edge $uv^i$ in
the graph $\hat G$.
Initially, we add to $H$ a key $(e,1)$ with value $w_e$ for all edges
$e\in{E}$ incident to $u$.
We iteratively extract from $H$ the key $(e=uv,i)$ with minimum value.  
Then we create an edge $uv^i$ in $E'$ with weight $i\cdot w_e$.
If $u$ is incident to less than $|U|$ edges in $E'$, then we insert to
$H$ a key $(uv,i+1)$ with value $(i+1)\cdot w_e$; otherwise, we stop.
We repeat the process until the heap $H$ is empty.
Thus, the process for each vertex $u\in U$ terminates in $|U|$
rounds.
The pseudocode of the reduction is given in
Algorithm~\ref{algo:reduction}.

\begin{algorithm}
\caption{\pname{Reduction} $(G=(U\cup V,E),w)$}
\label{algo:reduction}
\begin{algorithmic}[1]
  \STATE Create an empty set $\hat E$, $\hat V$.
  \FORALL{vertices $u\in U$}
    \STATE Create a binary heap $H$.
    \FORALL{edges $e$ incident to $u$}
       \STATE Insert to $H$ a node with key $(e,1)$ and value $w(u)$.
    \ENDFOR
    \FOR{$k\leftarrow 1$ to $|U|$}
       \STATE Extract-min from $H$, resulting in $(e=uv, i)$.
       \STATE Insert to $\hat V$ a vertex $v^i$
              (if it does not exist).
       \STATE Insert to $\hat E$ an edge $uv^i$.
       \STATE Insert to $H$ a node with key $(e=uv,i+1)$ and
              value $(i+1)\cdot w_e$.
    \ENDFOR
    \STATE Delete the binary heap $H$.
  \ENDFOR
  \RETURN $\hat G=(U\cup\hat V,\hat E)$.
\end{algorithmic}
\end{algorithm}

Consider a vertex $u\in U$. At any time during the reduction, there
are $O(\deg_G(u))$ edges in $H$. So, the extract-min operation takes
$O(\log(\deg_G(u)))$ time per operation. The time for inserting a vertex to $\hat
V$ and an edge to $\hat E$ is $O(1)$.
%
%
For each vertex $u\in U$, we have to call insertion $\deg_G(u)+|U|$ times and extract-min $|U|$ times.  Thus, the time
required to process each vertex of $U$ is
$O((\deg_G(u)+|U|)\log|U|)$. It follows that the total running time
of the reduction is $O((|E|+|U|^2)\log |U|)=O(n^2\log n)$.

Now, we run algorithms for the bipartite matching problem on the graph
$\hat{G}$ with $n^2$ edges.
Using Edmonds-Karp-Tomizawa algorithm, the running time becomes
$O(nm)=O(n^3)$.
Using  Gabow-Tarjan's scaling algorithm,
the running time becomes $O(\sqrt{n}m\log{(nW)}=O(n^{5/2}\log{(nW)})$,
where $W$ is the maximum edge weight.


\section{Dinitz's blocking flow algorithm}
\label{sec:Dinitz}

In this section, we will give an outline of Dinitz's blocking flow
algorithm~\cite{Dinitz70}. Given a network $R$ with source $s$ and
sink $t$, a flow $g$ is a {\em blocking flow} in $R$ if every path
from the source to the sink contains a saturated edge, an edge with
zero residual capacity.  A blocking flow is usually called a greedy
flow since the flow cannot be increased without any rerouting of
the previous flow paths.  In a unit capacity network, the depth-first 
search algorithm can be used to find a blocking flow in linear time. 

Dinitz's algorithm works in a {\em layered graph}, a subgraph whose
edges are in at least one shortest path from $s$ to $t$.  This
condition implies that we only augment along the shortest paths. The
algorithm proceeds by successively find blocking flows in the layered
graphs of the residual graph of the previous round. The following is
an important property (see, e.g.,
\cite{AMOBook,SchrijverA,TarjanBook} for proofs).  It states that
the distance between the source and the sink always increase after
each blocking flow step.

In the case of unit-capacity, Even-Tarjan~\cite{ET75} and
Karzanov~\cite{Karzanov73} showed algorithm that finds a maximum
flow in time $O(\min\{n^{2/3},m^{1/2}\}m)$. In the case of
{\em unit-network}, i.e., every vertex either has indegree $1$ or
outdegree $1$, the algorithm finds a maximum flow in time
$O(\sqrt{n}m)$.


\end{document}